\documentclass[letter]{article}
\usepackage[affil-it]{authblk}
\usepackage[margin=1.25in]{geometry}
\usepackage{amsmath}
\usepackage{algpseudocode}
\usepackage{graphicx}
\algtext*{EndFor}
\algtext*{EndIf}
\algtext*{EndWhile}
\usepackage{algorithm}
\usepackage{multirow}
\usepackage{amsthm}
\usepackage{framed}

\theoremstyle{definition}
\newtheorem{definition}{Definition}
\theoremstyle{plain}
\newtheorem*{prob}{Problem}
\newtheorem{obs}{Observation}
\newtheorem{theorem}{Theorem}
\newtheorem{lemma}{Lemma}
\newtheorem{corollary}{Corollary}

\usepackage{graphicx}
\usepackage{amssymb}
\DeclareGraphicsExtensions{.pdf}

\newcommand{\comment}[1]{}

\begin{document}

\title{DMVP: Foremost Waypoint Coverage of Time-Varying Graphs}

\author[1]{Eric Aaron\thanks{eaaron@cs.vassar.edu}}
\author[2]{Danny Krizanc\thanks{ekmeyerson@wesleyan.edu}}
\author[2]{Elliot Meyerson\thanks{dkrizanc@wesleyan.edu}}

\affil[1]{Computer Science Department, Vassar College, Poughkeepsie, NY, USA}
\affil[2]{Department of Mathematics \& Computer Science, Wesleyan University, Middletown, CT, USA}

\date{}

\maketitle

\begin{abstract}
We consider the Dynamic Map Visitation Problem (DMVP), in which a team of agents must visit a collection of critical locations as quickly as possible, in an environment that may change rapidly and unpredictably during the agents' navigation. We apply recent formulations of time-varying graphs (TVGs) to DMVP, shedding new light on the computational hierarchy $\mathcal{R} \supset \mathcal{B} \supset \mathcal{P}$ of TVG classes by analyzing them in the context of graph navigation. We provide hardness results for all three classes, and for several restricted topologies, we show  a separation between the classes by showing severe inapproximability in $\mathcal{R}$, limited approximability in $\mathcal{B}$, and tractability in $\mathcal{P}$. We also give topologies in which DMVP in $\mathcal{R}$ is fixed parameter tractable, which may serve as a first step toward fully characterizing the features that make DMVP difficult.
\end{abstract}

\section{Introduction}

In navigation-oriented application domains such as autonomous\comment{ vehicles,} mobile robots, wireless sensor networks, security, surveillance, mechanical inspection, and more, graph representations are commonly employed for formulating and analyzing the central\comment{ multi-agent} navigation or area inspection problems. Many approaches to coverage problems \cite{Choset01,Correll08,Easton05} are based on static graph representations, as are visitation problems \cite{MVP} or related combinatorial optimization problems such as the $k$-Chinese Postman Problem \cite{Ahr06,Blum94} and $k$-Traveling Repairman Problem \cite{Edmonds73,Rao07}. But static graph structures do not represent the\comment{ kinds of} dynamic environments that can occur in applications of autonomous robots or non-player characters in video games and virtual worlds. In this paper, we present the \emph{Dynamic Map Visitation Problem} (DMVP), applying recent formulations of \emph{highly dynamic graphs} (or \emph{time-varying graphs} (TVGs)) \cite{Santoro12,Kuhn11} to an essential graph navigation problem: In DMVP, a team of agents must inspect a collection of critical locations on a map (represented as a graph) as quickly as possible, but the agents' environment may change during navigation.

The application of TVG models is essential to DMVP. In applications such as planetary exploration \cite{Wagner99}, search and rescue in hazardous environments (e.g., natural disasters, areas of armed conflict), or even ad-hoc network inspection, many aspects of the structure of graph waypoints and edges governing navigation can change during agent navigation, and TVG models can capture variation in graph structure in ways that static graphs cannot. Our paper presents new results about DMVP complexity and demonstrates distinctions among classes of TVGs; details of our main results are summarized in Section~\ref{ourResults}.

When incorporating dynamics into a problem such as DMVP, there are many options for how to constrain/model the dynamics of the graph. Dynamics can be deterministic (e.g., \cite{Santoro10,Wade11,Wade13,Mans13,Xuan03}) or stochastic (e.g., \cite{Baumann11,Lotker08}). In this paper, to provide a foundation for future work, and exemplify the aspects of topologies and dynamics that make our problem easy or hard, we focus on the deterministic case. The deterministic approach is also particularly relevant for situations in which some prediction of changes is feasible. Quite a bit of this previous work has required that the graph be connected \emph{at all times} \cite{Lotker08,Wade13,Kuhn10}. Indeed, for complete map visitation to be possible, every critical location must be eventually reachable. However, in application environments such as those outlined above, at any given time the waypoint graph may be disconnected. Our model must be general enough to allow for this phenomenon. 

We adopt three classes of TVGs, each of which places constraints on edge dynamics. In $\mathcal{R}$, edges must reappear eventually; in $\mathcal{B}$, edges must appear within some time bound; in $\mathcal{P}$ edge appearances are periodic. These classes have proven to be critical to the TVG taxonomy \cite{Santoro12}. They have been studied with respect to problems such as broadcast \cite{Santoro10} and exploration \cite{Santoro13,Wade11}, with results relating to feasibility of computation and bounds on broadcast and exploration time. $\mathcal{R}$, $\mathcal{B}$, and $\mathcal{P}$ place intuitive constraints on the nature of dynamic navigation domains. Even the assumption of periodicity of edges has applications to navigation of transportation networks \cite{Santoro13,Wade11}, as well as environments periodically patrolled by other agents, who can prohibit or guarantee safe traversal of an edge.

In this paper, we shed further light on the computational hierarchy of $\mathcal{R}$, $\mathcal{B}$, and $\mathcal{P}$ \cite{Santoro10}, by analyzing them in the context of DMVP, a natural but difficult problem in global navigation. We provide hardness results for all three classes. For several restricted topologies, we demonstrate separation between the classes by showing severe inapproximability in $\mathcal{R}$, limited approximability in $\mathcal{B}$, and tractability in $\mathcal{P}$. We also give topologies in which DMVP in $\mathcal{R}$ is tractable and fixed parameter tractable, which may serve as a first step towards fully characterizing the topological features that make DMVP difficult. Because our goal in this paper is to cleanly differentiate the classes of dynamics we are exploring, rather than explore the interactions between multiple agents, our results here focus on the case of a single agent.

\subsection{Definitions and TVG Concepts}
\label{definitions}

As a foundation for our work, we adopt the definitions below from Santoro et al. \cite{Santoro12}.

\begin{definition}
A TVG (time-varying graph, dynamic graph, or dynamic network) is a five-tuple $\mathcal{G} = (V, E, \mathcal{T}, \rho, \zeta)$, where $\mathcal{T} \subseteq \mathbb{T}$ is the \emph{lifetime} of the system, \emph{presence function} $\rho(e,t) = 1 \iff$ edge $e \in E$ is available at time $t \in \mathcal{T}$, and \emph{latency function} $\zeta(e,t)$ gives the time it takes to cross $e$ if starting at time $t$. The graph $G = (V,E)$ is called the \emph{underlying graph} of $\mathcal{G}$, with $\lvert V \rvert = n$.
\end{definition}

In the most general case, $\mathbb{T}$ can be $\mathbb{R}$, and edges can be directed. However, in our work we consider the discrete case in which $\mathbb{T} = \mathbb{N}$, edges are undirected, and all edges have uniform travel cost $\zeta(e,t) = 1$ at all times. If agent $a$ is at $u$, and edge $(u,v)$ is available at time $\tau$, then $a$ can take $(u,v)$ during this time step, visiting $v$ at time $\tau+1$. As $a$ traverses $\mathcal{G}$ we say $a$ both \emph{visits} and \emph{covers} the vertices in its traversal, and we will henceforth use these terms interchangeably.  A \emph{temporal subgraph} of a TVG $\mathcal{G}$ results from restricting the lifetime $\mathcal{T}$ of $\mathcal{G}$ to some $\mathcal{T}' \subseteq \mathcal{T}$. A \emph{static snapshot} is a temporal subgraph throughout which the availability of each edge does not change, i.e., edges are static.

\begin{definition}
$\mathcal{J} = \{(e_1,t_1),...,(e_k,t_k)\}$ is a \emph{journey} $\iff \{e_1,...,e_k\}$ is a walk in $G$ (called the \emph{underlying walk} of $\mathcal{J}$), $\rho(e_i,t_i) = 1$ and $t_{i+1} \geq t_i + \zeta(e_i,t_i)$ for all $i < k$. The \emph{topological length} of $\mathcal{J}$ is $k$, the number of edges traversed. The \emph{temporal length} is the duration of the journey: $(arrival \ date) - (departure \ date)$. 
\end{definition}

Given a date $t$, a journey from $u$ to $v$ departing on or after $t$ whose arrival date $t'$ is soonest is called \emph{foremost}; whose topological length is minimal is called \emph{shortest}; and whose temporal length is minimal is called \emph{fastest}.

In \cite{Santoro12}, a hierarchy of thirteen classes of TVG's is presented. In related work on exploration \cite{Santoro13} and broadcast \cite{Santoro10}, focus is primarily on the chain $\mathcal{R} \supset \mathcal{B} \supset \mathcal{P}$ defined below. We adopt these classes into our domain, which we believe enforce natural constraints in our application environments. 

\begin{definition}{(Recurrent edges)}  $\mathcal{R}$ is the class of all TVG's $\mathcal{G}$ such that $G$ is connected, and $ \forall e \in E, \forall t \in \mathcal{T}, \exists t' > t$ s.t. $ \rho(e,t') = 1$.
\end{definition}

\begin{definition}{(Time-bounded recurrent edges)}  $\mathcal{B}$ is the class of all TVG's $\mathcal{G}$ such that $G$ is connected, and $ \forall e \in E, \forall t \in \mathcal{T}, \exists t' \in [t, t + \Delta)$ s.t. $\rho(e,t') = 1$, for some integer $\Delta$. 
\end{definition}

\begin{definition}{(Periodic edges)}  $\mathcal{P}$ is the class of all TVG's $\mathcal{G}$ such that $G$ is connected, and $ \forall e \in E, \forall t \in \mathcal{T}, \forall k \in \mathbb{N}, \rho(e,t) = \rho(e,t + kp)$ for some integer $p$. $p$ is called the \emph{period} of $\mathcal{G}$.
\end{definition}

As much as possible, we also take standard notation and terms from the graph theory literature. We rely on several underlying graph topologies. A \emph{star} is a tree in which at most one vertex has degree greater than one. The leaves of a star are called \emph{points}. A \emph{spider} is a tree in which at most one vertex has degree greater than two. In other words, a spider consists of a set of vertex-disjoint paths, called \emph{arms}, each of which has exactly one endpoint connected to the common central vertex $c$. A \emph{comb} is a max-degree 3 tree, in which there exists a simple path containing every vertex of degree 3. Such a path is called a \emph{backbone} of the comb. Paths edge-disjoint to the backbone are called \emph{arms}. A leaf distance 1 from the backbone is called a \emph{tooth}. An \emph{r-almost-tree} is a connected graph with $\lvert V \rvert + r - 1$ edges, that is, $r$ edges can be removed to produce a tree.

\begin{prob}
Given a TVG $\mathcal{G}$ and a set of starting locations $S$ for $k$ agents in $G$, the TVG foremost coverage or dynamic map visitation problem (\emph{DMVP}) is the task of finding journeys starting at time 0 for each of these $k$ agents such that every node in $V$ is in some journey, and the maximum temporal length among all $k$ journeys is minimized. The decision variant asks whether these journeys can be found such that no journey ends after time $t$.
\end{prob}

We think of the input $\mathcal{G}$ as a temporal subgraph of some TVG $\mathcal{G}_\infty$ with lifetime $\mathbb{N}$ and the same edge constraints as $\mathcal{G}$. Thus, the limited information provided in $\mathcal{G}$ is used to compute complete solutions for agents covering $\mathcal{G}_\infty$. When unspecified, assume that DMVP refers to DMVP for a single agent.

\subsection{Main Results}
\label{ourResults}

Our results are summarized in Table~\ref{resultsTable}. We show that DMVP in $\mathcal{R}$ is NP-hard to approximate within any factor, when the underlying graph $G$ is restricted to a star or tree of max degree 3. We show that in $\mathcal{B}$ this problem is NP-hard to approximate within any factor less than $\Delta$, when $G$ is restricted to a spider or tree of max degree 3. We show that in $\mathcal{P}$, DMVP is NP-complete when $p=1$, and that there is a nontrivial class of graphs for which $p=2$ is NP-hard, but $p=1$ is not.

We show that in $\mathcal{R}$, DMVP is solvable in $O(T)$ when $G$ is a path, $O(Tn)$ when $G$ is a cycle, and $O(Tn^3+n^22^n)$ for general graphs,  where $T$ is the duration of $\mathcal{G}$, as defined in Section \ref{modelAndProblem}. Furthermore, in $\mathcal{R}$, DMVP is fixed parameter tractable when $G$ is an $m$-leaf $O(1)$-almost tree, and poly-time solvable when $m = O(\lg n)$. In $\mathcal{B}$, we demonstrate a tight $\Delta$-approximation for trees, and a $2\Delta$-approximation for general graphs. We demonstrate a class of problems which are NP-hard in $\mathcal{B}$, but solvable by an online algorithm in $\mathcal{P}$. We show that DMVP in $\mathcal{P}$ is solvable in polynomial time when $G$ is a spider, for fixed $p$, and we show that when $p=2$, DMVP is solvable in linear time for general trees.

\begin{table}[h]
\caption{DMVP separations and results by TVG class and graph class\label{resultsTable}}
\begin{center}
\resizebox{12.5cm}{!} {
\begin{tabular}{| c | c | c | c |}
\hline
\multicolumn{4}{ | c | }{DMVP separations} \\
\hline
TVG class & spiders & max-degree 3 trees & general trees\\\hline
$\mathcal{R}$ & no approx. & no approx. & no approx. \\\hline
$\mathcal{B}$ & tight $\Delta$-approx. & tight $\Delta$-approx. & tight $\Delta$-approx. \\\hline
$\mathcal{P}$ & in P, for fixed $p$ & $O(n)$ exact, for $p=2$ & $O(n)$ exact, for $p=2$ \\\hline
\multicolumn{4}{ | c | }{$\exists$ graph class over which DMVP NP-hard in $\mathcal{P}$ with $p=2$, easy with $p=1$.} \\\hline
\end{tabular}}

\resizebox{12.5cm}{!} {
\begin{tabular}{| c | c | c | c | c |}
\hline
\multicolumn{5}{ |c| }{Complexity of exact algorithms in $\mathcal{R}$} \\
\hline
path & cycle & general graphs & $m$-leaf $c$-almost trees & $O(\lg n)$-leaf $c$-almost trees \\\hline
$O(T)$ & $O(Tn)$ & $O(Tn^3+n^22^n)$ & in FPT & in P \\\hline
\end{tabular}}
\end{center}
\end{table}

The remainder of this paper is organized as follows: preliminaries (\ref{modelAndProblem}), lower bounds (\ref{lowerBounds}), upper bounds (\ref{upperBounds}), open problems and discussion (\ref{openProbs}).

\section{Preliminaries}
\label{modelAndProblem}

For the minimization problem $\mbox{\emph{DMVP}}(\mathcal{G},S)$ and the corresponding decision problem $\mbox{\emph{DMVP}}(\mathcal{G},S,t)$, input is viewed as a sequence of graphs $G_i$ each represented as an adjacency matrix, with an associated integer duration $t_i$, i.e. $\mathcal{G} = (G_1,t_1),(G_2,t_2),$ $...,(G_m,t_m)$, where $G_1$ appears initially at time zero. Let $T = \sum_{i=1}^{m} t_i$. Note that since each $t_i$ can be encoded in $O(\lg t_i)$ space, it is possible for $T$ to be exponential in the size of $\mathcal{G}$. The following observation is required to show that the number of time steps of $\mathcal{G}$ that need to be considered for DMVP is in fact polynomial in the size of $\mathcal{G}$.

\begin{obs}
\label{2n}
When computing DMVP over $\mathcal{G}$, it is not necessary to consider each static temporal subgraph $(G_i,t_i)$ for more than $2n-3$ time steps. 
\end{obs}
\begin{proof}
Suppose $G_i$ is the available static subgraph of $\mathcal{G}$ from times $\tau$ to $\tau + t_i$, and $t_i > 2n-3$. Suppose agent $a$ starts at vertex $u$ at time $\tau$. There are two cases:

Case 1:  If $a$ can complete its coverage of $\mathcal{G}$ by only traversing in $G_i$, then in the worst case $a$ can execute any complete spanning tree traversal of $G_i$, which takes no more than $2n-3$ steps. In this case, it does not matter at which vertex $a$ ends up, because the task has been completed.

Case 2: If there is a vertex $v$ such that $a$ has not covered $v$ by time $\tau$, and $u$ and $v$ are in different connected components in $G_i$, then $a$ cannot complete coverage of $\mathcal{G}$ when $G_i$ is the available static subgraph. In this case it may matter which vertex $a$ ends up at,  depending on which future edges will be available. The size of the connected component of $u$ in $G_i$ is at most $n-1$, so a spanning tree traversal of this component ending up back at $u$ takes no more than $2n-4$ steps. If $a$ would rather end up at a different vertex $w \neq u$, it simply traverses $w$'s branch of the spanning tree last, and returns up only to $w$, in fewer than $2n-4$ steps. 
\end{proof}

By Observation \ref{2n}, for any $t_i > 2n-3$, when computing DMVP, all time steps after the first $2n-3$ can be ignored (skipped). DMVP over $\mathcal{G}$ can be computed by computing DMVP over $\mathcal{G}' = (G_1,\min(t_1,2n-3)),...,(G_m,\min(t_m,2n-3))$, and adding back the cumulative time skipped before completion. $\mathcal{G}'$ can clearly be derived from $\mathcal{G}$ in $O(m)$ time. The total duration of $\mathcal{G}'$ is $T' = \sum_{i=1}^{m} \min(t_i,2n-3) < 2nm - 3m$, which is polynomial in $\lvert \mathcal{G} \rvert$. Let $\epsilon(\tau)$ be the time skipped through time $\tau \leq T'$. $\epsilon(\tau)$ can be simply calculated for all $\tau \leq T'$ in $O(T')$ time. A similar $O(T')$ preprocessing step can be run to associate each  time $\tau \in T'$ with the corresponding available static graph $G_i$, enabling $O(1)$ edge presence lookups $\rho(e,\tau)$.

Since all of the algorithms we present run in $\Omega(T')$ time, we can run these preprocessing steps for every instance of DMVP and not affect the asymptotic running time. Therefore, for the sake of simplicity, for the rest of our results we assume that this preprocessing has taken place, i.e., we think of $\mathcal{G}$ as $\mathcal{G'}$ and $T$ as $T'$, thereby avoiding the exponential nature of $T$. Note also that for the case of $\mathcal{P}$, the constraint of periodicity implies that it is only necessary to look at $p$ consecutive time steps of the input.

\section{Lower Bounds}
\label{lowerBounds}

As motivation for many of the results in this paper, it is important to note that MVP for a single agent is solvable in linear time on trees \cite{MVP}. To characterize the difficulty of DMVP in $\mathcal{R}$, we first show inapproximability over stars. A similar theorem was independently discovered in \cite{Michail14}.

\begin{theorem}
\label{npostar}
DMVP for a single agent in $\mathcal{R}$ is NP-hard to approximate within any factor, even when the underlying graph is a star.
\end{theorem}
\begin{proof}
We reduce from the set cover problem (SCP). Given a universe $\mathcal{U} = \{1,2,...,m\}$, a family $\mathcal{S}$ of subsets $s_1,s_2,...,s_n$ of $\mathcal{U}$, and an integer $k$, it is NP-complete to decide whether $\mathcal{S}$ contains a cover of $\mathcal{U}$ of size $k$ or less \cite{Karp}. Given an instance of SCP, construct a star $G = (V,E)$ with central vertex $c$; points $v_1,...,v_m$, corresponding to elements in $\mathcal{U}$; $p_1,...,p_n$, corresponding to sets in $\mathcal{S}$; and a single check point $p_0$. We use the following static subgraphs to construct a TVG $\mathcal{G}$. For all $s_i$ in $\mathcal{S}$, let $pass(i) = (V, \{c,p_i\})$, and let $take(i) = (V, E_i)$, where $E_i = \{(c,v_j) : j \in s_i \}$. Let $check = (V,\{c,p_0\})$. Let $finish = (V, F)$, where $F = \{(c,p_i) : 1 \leq i \leq n\}$. Consider the TVG $\mathcal{G}  = (pass(1),1)$,$(take(1),t_1)$,$(pass(1),1)$,...,$(pass(n),1)$,$(take(n),t_n)$,$(pass(n),1)$,$(check,2)$,$(finish,2k-1)$, where $t_i = 2\lvert s_i \rvert, \forall i \in \{1,...,n\}$. The total duration of $\mathcal{G}$ is $D = 2n + 2\sum_{i=1}^n \lvert s_i \rvert + 2k + 1 < 2n + 2mn + 2k + 1$. 
	
Consider the problem of deciding if DMVP over $\mathcal{G}$ with a single agent $a$ starting at $c$ has a solution of length no more than $D$. This problem is in NP, since given a journey over $\mathcal{G}$, we can easily check that it hits every vertex, and that all of its edges are available at the correct times. Such a solution can be no longer than $D$, since $D$ is the total duration of $\mathcal{G}$.

Suppose $\mathcal{S}$ contains a cover $\mathcal{C}$ of $\mathcal{U}$ of size $k$ or less. Then for all $s_i \in \mathcal{C}$, $a$ $takes$ $s_i$, that is, $a$ waits at $c$ during both instances of $pass(i)$, and visits all $v_j \in s_i$ and returns to $c$ during $take(i)$, which is possible since the duration of $take(i)$ is $2\lvert s_i \rvert$. Since $\mathcal{C}$ is a cover of $\mathcal{U}$, $a$ visits all $v_i$. For all $s_i \notin \mathcal{C}$, $a$ $passes$ $s_i$, that is, $a$ moves from $c$ to $p_i$ during the first $pass(i)$, waits at $p_i$ during $take(i)$, and returns to $c$ during the second $pass(i)$. During $check$, $a$ moves from $c$ to $p_0$, then back to $c$. At this point, since $\lvert \mathcal{C} \rvert \leq k$, $a$ has passed at least $n-k$ $s_i$'s. So, there are no more than $k$ $p_i$'s left unvisited. $a$ visits these during $finish$, thus completing visitation of all vertices of $G$ in no more than $D$ steps (e.g., Figure~\ref{RstarExample}).

Suppose there exists a solution to this instance of DMVP of length no more than $D$. Prior to $finish$, $a$ must have visited at least $n-k$ $p_i$'s, since $finish$ only lasts for $2k-1$ steps. So $a$ must have passed at least $n-k$ $s_i$'s. Taking and passing for a single $s_i$ are mutually exclusive, because if $a$ moves to $p_i$ during the first $pass(i)$, $a$ must wait during $take(i)$, and if $a$ both takes $s_i$ and moves to $p_i$ during the second $pass(i)$, $a$ will be trapped at $p_i$ until $finish$, and will never be able to reach $p_0$, which must be visited during $check$, the only input time at which $p_0$ is available. Thus, $a$ could have taken no more than $k$ $s_i$'s. During these $k$ or fewer takes, $a$ must have covered all $v_1,...,v_m$. So, the union of these $k$ or fewer $take(i)$'s contains all edges $(c,v_j)$, which implies that the corresponding $s_i$'s form a cover of $\mathcal{U}$ or size $k$ or less. Hence, the decision problem is NP-complete.  

Consider the minimization version of the problem with the same setup. Since it is NP-hard to decide if there is a solution of length $D$ or less, it is NP-hard to find such a solution. But after $D$ steps, $a$ may have to wait an arbitrarily long time for the next edge is a feasible solution to appear, so any feasible solution that takes longer than $D$ steps can be arbitrarily long. Therefore, the problem cannot be approximated within any factor.
 \end{proof}
 
 \begin{figure}
\includegraphics[scale=1,trim=4cm 4cm 4cm 4cm]{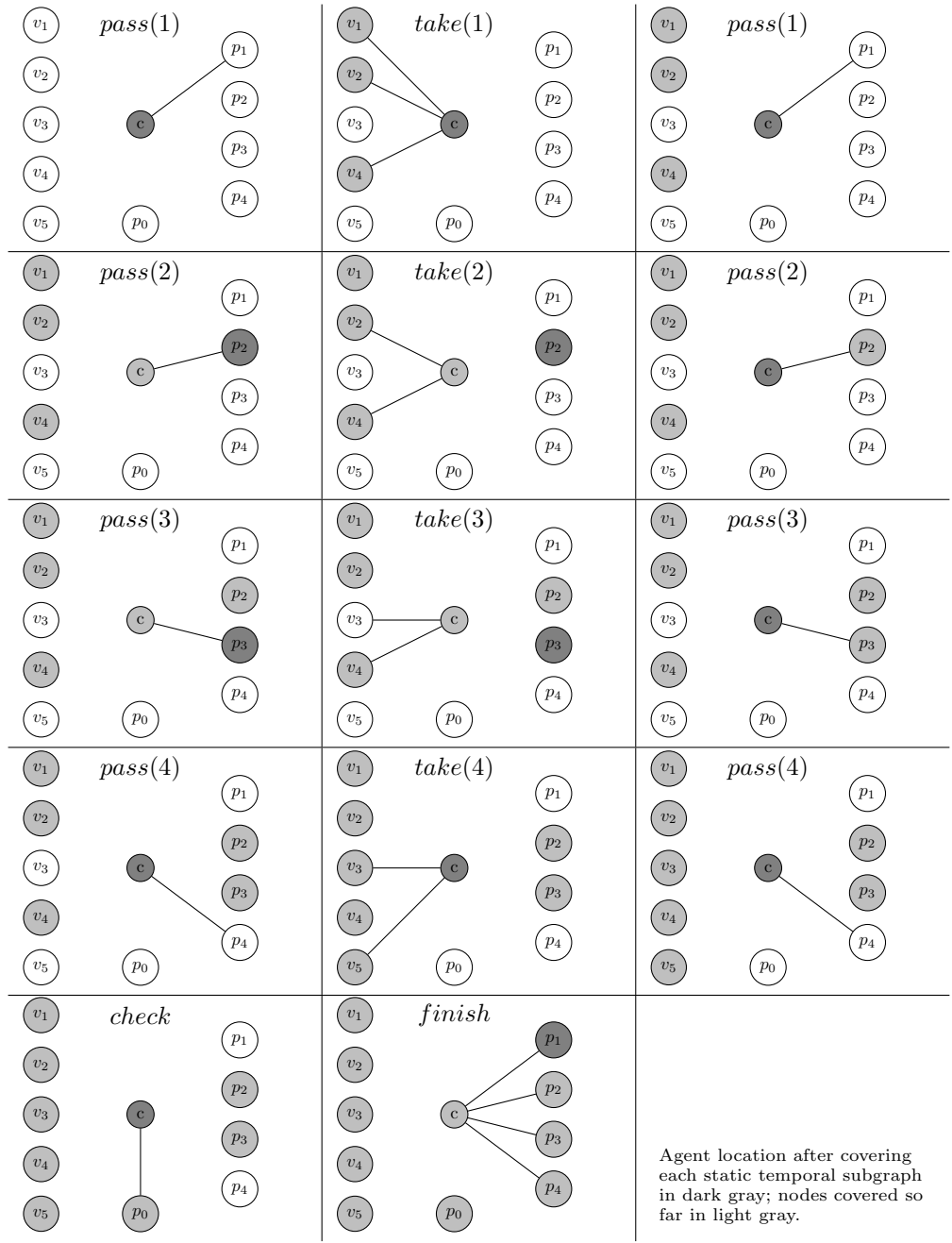}
\caption{Thm.\ref{npostar} static snapshots for $U = \{1,2,3,4,5\}$, $S =  \{\{1,2,4\},\{2,4\},\{3,4\},\{3,5\}\}$, $k = 2$.\label{RstarExample}}
\end{figure}

This inapproximability also holds over the restriction of underlying graphs to trees of max-degree 3, in particular, combs.

\begin{theorem}
\label{npocomb}
DMVP for a single agent in $\mathcal{R}$ is NP-hard to approximate within any factor, even when the underlying graph is a comb.
\end{theorem} 
\begin{proof}
Analogous to Theorem \ref{npostar}, we reduce from the set cover problem (SCP) \cite{Karp}. Given an instance of SCP, construct a comb $G = (V,E)$ with backbone $b_0b_1b_{m+n+1}$; teeth $v_1,...,v_m$, corresponding to elements in $\mathcal{U}$, with $(b_i,v_i) \in E \ \forall i = 1,...,m$; teeth $p_1,...,p_n$, corresponding to sets in $\mathcal{S}$, with $(b_{m+i},p_i) \in E \ \forall i = 1,...,n$; and two check teeth $p_0$ and $p_{n+1}$ with $(b_0,p_0),(b_{m+n+1},p_{n+1}) \in E$. Let $B = \{(b_0,b_1),...,(b_{m+n-1},b_{m+n})\}$ be the set of all edges in the backbone of $G$. We use the following static subgraphs of $G$ to construct a TVG $\mathcal{G}$. For all $s_i$ in $\mathcal{S}$, let $pass(i) = (V, B \cup \{(b_{m+i},p_i)\})$, and let $take(i) = (V, E_i)$, where $E_i = B \cup \{(b_j,v_j) : j \in s_i \}$. Let $check = (V, \{(b_0,p_0)\})$. Let $finish = (V, F)$, where $F = B \cup \{(b_{m+i},p_i) : 1 \leq i \leq n+1\} \cup \{(b_{m+n+1}, p_{n+1})$. Let $back = (V, B)$.

Define the TVG $\mathcal{G}  = (back, m+n),(pass(1),1),(take(1),3m),(pass(1),1),...$,$(back, m+n)$,\\$(pass(n),1)$,$(take(n),3m)$,$(pass(n),1),(back, m+n),(check,2),$$(finish,m+n+2+2k)$. The total duration of $\mathcal{G}$ is $D = n(4m+n+2)+(m+n)+2+(m+n+1+2k) = n^2 + 4mn + 4n + 2m + 2k + 3$.
	
Consider the problem of deciding if DMVP over $\mathcal{G}$ with a single agent $a$ starting at $b_0$ has a solution of length no more than $D$. This problem is in NP, since given a journey over $\mathcal{G}$, we can easily check that it hits every vertex, and that all of its edges are available at the correct times. Such a solution can be no longer than $D$, since $D$ is the total duration of $\mathcal{G}$.

Suppose $\mathcal{S}$ contains a cover $\mathcal{C}$ of $\mathcal{U}$ of size $k$ or less. Then for all $s_i \in \mathcal{C}$, $a$ $takes$ $s_i$, that is, $a$ travels to $b_0$ during \emph{back}, and visits all $v_j \in s_i$, and returns to the backbone during $take(i)$, which is possible since the duration of $take(i)$ is $3m$, which allows $a$ to take all of the at most $m$ available elements while traveling up the backbone. Since $\mathcal{C}$ is a cover of $\mathcal{U}$, $a$ visits all $v_i$. For all $s_i \notin \mathcal{C}$, $a$ $passes$ $s_i$, that is, $a$ moves to $b_{i+m}$ during $back$, and to $p_i$ during the first $pass(i)$, waits at $p_i$ during $take(i)$, and returns to $b_{i+m}$ during the second $pass(i)$. During the final $back$, $a$ moves to $b_0$, and during $check$, $a$ moves from $b_0$ to $p_0$, then back to $b_0$. At this point, since $\lvert \mathcal{C} \rvert \leq k$, $a$ has passed at least $n-k$ $s_i$'s. So, there are no more than $k$ $p_i$'s left unvisited. $a$ visits these during $finish$, each at cost 2 off the path length $m+n+2$ path up the backbone to $p_{n+1}$, thus completing visitation of all vertices of $G$ in no more than $D$ steps (e.g., Table 1).

Suppose there exists a solution to this instance of DMVP of length no more than $D$. Prior to $finish$, $a$ must have visited at least $n-k$ $p_i$'s, since $finish$ only lasts for $2k-1$. So $a$ must have passed at least $n-k$ $s_i$'s. Taking and passing for a single $s_i$ are mutually exclusive, because if $a$ moves to $p_i$ during the first $pass(i)$, $a$ must wait during $take(i)$, and if $a$ both takes $s_i$ and moves to $p_i$ during the second $pass(i)$, $a$ will be trapped at $p_i$ until $finish$, and will never be able to reach $p_0$, which must be visited during $check$, the only input time at which $p_0$ is available. Thus, $a$ could have taken no more than $k$ $s_i$'s. During these $k$ or fewer takes, $a$ must have covered all $v_1,...,v_m$. So, the union of these $k$ or fewer $take(i)$'s contains all edges $(c,v_j)$, which implies that the corresponding $s_i$'s form a cover of $\mathcal{U}$ or size $k$ or less. Hence, the decision problem is NP-complete.  

Consider the minimization version of the problem with the same setup. Since it is NP-hard to decide if there is a solution of length $D$ or less, it is NP-hard to find such a solution. But after $D$ steps, $a$ may have to wait an arbitrarily long time for the next edge is a feasible solution to appear, so any feasible solution that takes longer than $D$ steps can be arbitrarily long. Therefore, the problem cannot be approximated within any factor.
\end{proof}

\begin{figure}
\includegraphics[scale=1,trim=4cm 4cm 4cm 4cm]{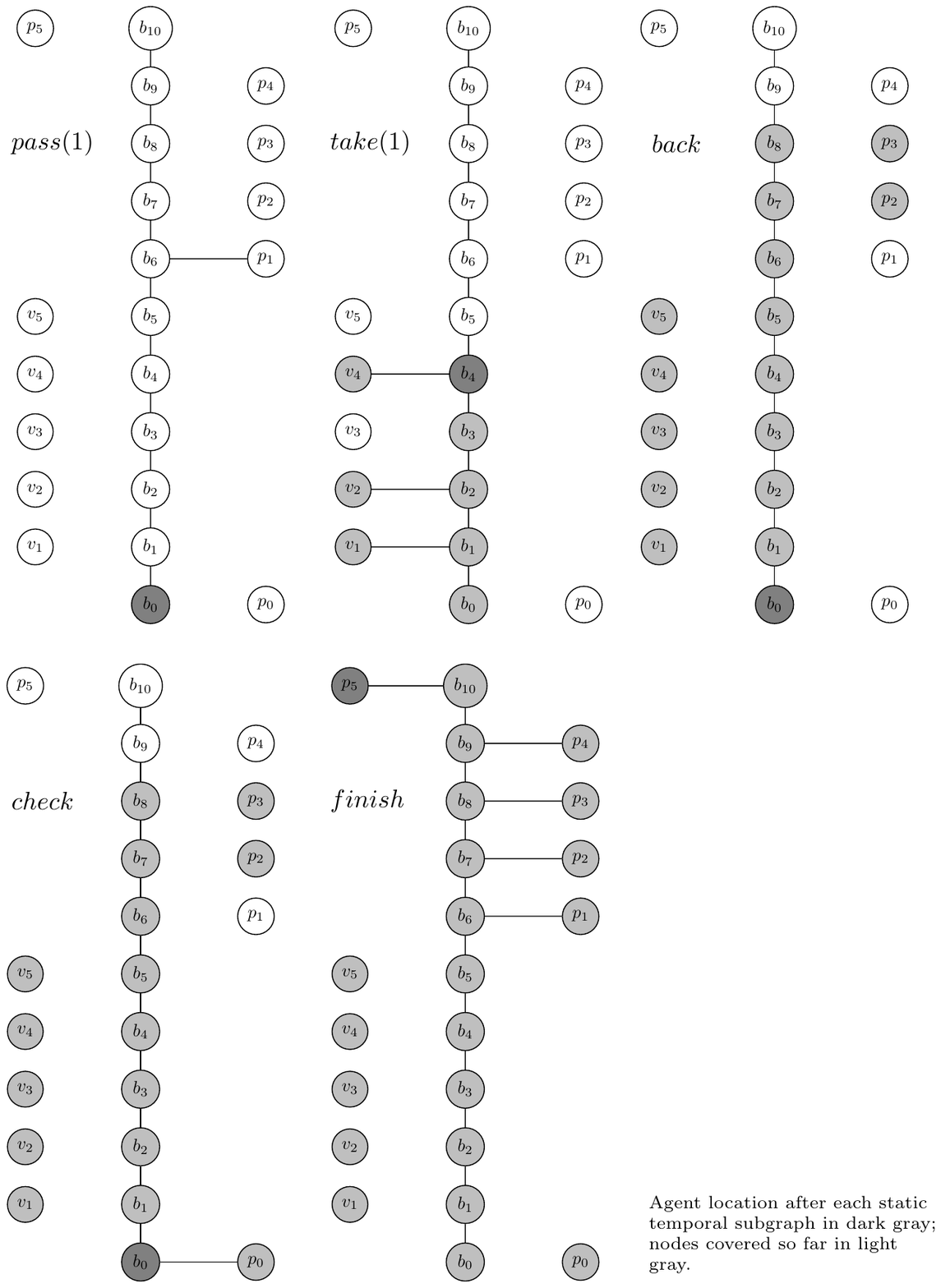}
\caption{Thm. 2 sample static snapshots for $U = \{1,2,3,4,5\}, s_1 = \{1,2,4\} \in S$, with $\lvert S \rvert = 4$, and $k = 2$.}
\end{figure}

We have a similar set of lower bounds for the case of $\mathcal{B}$, but with \emph{some} ability to approximate. We later show (Theorem \ref{btree}) that these approximation bounds are indeed tight for all trees. 

\begin{theorem}
\label{bapx}
DMVP for a single agent in $\mathcal{B}$ is NP-hard to approximate within any factor less than $\Delta$, even when the underlying graph is a spider, $\forall \Delta > 1$.
\end{theorem}
\begin{proof}
We reduce from 3-partition. Given a multiset of $3m$ positive integers $S = \{s_1,...,s_{3m}\}$, it is strongly NP-complete to decide if they can be partitioned into $m$ sets where all sets have the same sum \cite{Garey}. Let $\sum_{i=1}^{3m} s_i$ = M. Then $B = M/m$ is the required sum for each partition.

Starting with the common central vertex $c$, construct a spider in the following way. For each $s_i \in S$, add a corresponding arm of length $s_i$. Add $m$ arms of length one to be used as checkpoints, and add a single long arm $A_0$ of some length $k > 2M + 2m$ (e.g., Figure \ref{bSpiderExample}). For the TVG used in this proof, arms over any period of time are in one of three modes: $steady$, $flashing$, or $carrying$. When an arm $A$ is steady over a period of time from $\tau$ to $\tau'$, all of its edges are constantly available during that period. When $A$ is flashing, all of its edges synchronously alternate between being unavailable for $\Delta-1$ steps, and available for 1 step, satisfying the time-bounded recurrence constraint. Formally, if $A$ is flashing, then $\forall e \in A, t \in [\tau,\tau'],$
$$
 \rho(e,t) =
  \begin{cases}
   1 & \text{if } t-\tau \equiv p-1 \bmod p, \\
   0 & \text{otherwise.}
  \end{cases}
$$
When $A$ is carrying, all of its edges act as if $A$ were flashing, with the exception that the edge distance $i$ from $c$ is always available at time $\tau + i$, so that starting at time $\tau$, an agent at $c$ can travel along $A$ for $\tau' - \tau$ steps without waiting. Formally, if $A = ca_0...a_l$ is carrying, then $\forall (u,v) \in E(A), t \in [\tau,\tau'],$
$$
 \rho(e,t) =
  \begin{cases}
   1 & \text{if } t-\tau \equiv p-1 \bmod p, \\
   1 & \text{if } v = a_{t-\tau},\\
   0 & \text{otherwise.}
  \end{cases}
$$

Let $take$ be the temporal subgraph of duration $2B$ in which the arms corresponding to $s_i$'s are steady, and all others are flashing. Let $check$ be the temporal subgraph of duration 2 in which all checkpoint arms are steady, and all others are flashing. Let $finish$ be the temporal subgraph of duration $k$ in which all arms are flashing except for $A_0$, which is carrying. Let $\mathcal{G}$ be the TVG formed by alternating between $take$ and $check$ $m$ times, before ending with $finish$. The total duration of $\mathcal{G}$ is $D = 2M + 2m + k$.

Consider the problem of deciding if DMVP over $\mathcal{G}$ with a single agent $a$ starting at $c$ has a solution of length no more than $D$. Since $k > 2M + 2m$, such a solution must take the long arm last, as traversing this arm twice would result in a solution of length greater than $2k > (2M + 2m) + k = D$. Furthermore, since every arm must be traversed twice except the long arm, the topological length of a solution journey can be no less than $2M + 2m + k = D$. So a solution of temporal length no more than $D$ cannot involve waiting.

Suppose there exists a 3-partition of $S$. During each $take$, $a$ can explore the arms of the spider corresponding to one partition, and return to $c$ in exactly $2B$ steps. During the subsequent $check$, $a$ visits one checkpoint arm, and returns to $c$ in the allotted 2 steps. Repeating this process for the remaining $takes$ and $checks$, $a$ will cover all the $s_i$ arms and checkpoint arms without ever waiting, and end up again at $c$. Finally, during $finish$, $a$ takes the long arm $A_0$, reaching its leaf without waiting, completing coverage in $D$ steps. 

Now, suppose this instance of DMVP has a solution of length $D$. To avoid waiting, $a$ must take complete arms and return to $c$ during each $take$, so that it is not stalled by flashing. Similarly, $a$ must take a single checkpoint arm and return to $c$ during each $check$. Doing this $m$ times, $a$ has effectively partitioned the $s_i$ arms into sets each of total length $B$. So, the decision problem is NP-complete, since 3-partition remains NP-complete even when input integers are given in unary. $a$ completes the solution by following the long arm $A_0$, which can be traversed in $k$ steps immediately after $a$ returns to $c$ from the final checkpoint. 

Consider the minimization version of the problem with the same setup. Note that if $a$ does not begin taking $A_0$ right as $finish$ begins, $A_0$ will take at least $\Delta(k-1)+1$ to traverse, this best case occurring when $a$ does not have to wait for the first edge. In particular, if $a$ takes $A_0$ last, but has not visited all other arms before $finish$ starts, visiting those arms must have taken at least $2M + 2m + (\Delta - 1) > 2M + 2m, \forall \Delta > 1$, since $a$ must wait at least once during its traversal. The total cost of the solution is then at least $D' = \Delta(k-1) + 1 + 2M + 2m + (\Delta - 1) = \Delta k + 2M + 2m$. If $a$ does not take $A_0$ last, it must traverse $A_0$ twice, taking at least $\Delta(k-1) + 1 + k$ steps (this best case occurring when $a$ starts taking the long arm out right as $finish$ starts), and once it returns must wait at least once while traversing the remaining arms, making the length of the total solution at least $D'' = \Delta(k-1) + 1 + k + 2M + 2m + (\Delta - 1) = \Delta k + 2M + 2m + k > D'$. Take any real constant $\delta < \Delta$. Choose the least integer $N$ s.t. $N > \frac{1}{\Delta-\delta}$. Let $k = N\delta(2M+2m)$. Then,
				$$(\Delta-\delta)N\delta(2M+2m) > \delta(2M+2m),$$
				$$(\Delta-\delta)k > \delta(2M+2m),$$
				$$ \Delta k > \delta(2M + 2m + k),$$
				$$ \Delta k + 2M + 2m = D' > \delta D, \ \forall \Delta > 1.$$
Therefore, any solution that contains waiting cannot be a $\delta$-approximation. So, finding a $\delta$-approximation is equivalent to finding a solution with no waiting, i.e., a minimal solution, and thus is NP-hard. Hence, the problem is NP-hard to approximate within any factor less than $\Delta$. 
 \end{proof}
 
\begin{figure}
\centering
\includegraphics[scale=1]{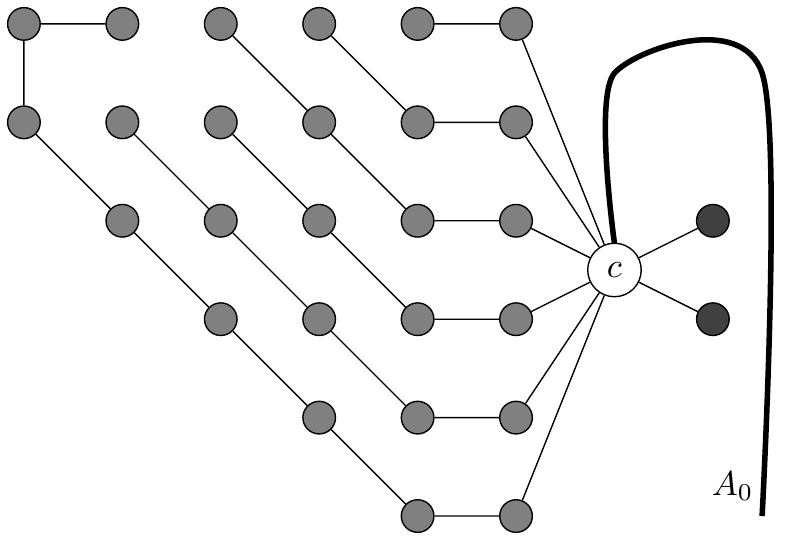}
\caption{Underlying spider for Thm. \ref{bapx} for 3-partition input: $S = \{2,3,4,4,5,8\}$. \label{bSpiderExample}}
\end{figure}

\begin{theorem}
\label{bcomb}
DMVP for a single agent in $\mathcal{B}$ is NP-hard to approximate within any factor less than $\Delta$, even when the underlying graph is a comb, $\forall \Delta > 1$.
\end{theorem}

\begin{proof}
We use a similar extension to that for $\mathcal{R}$ to extend this result from spider to a comb with \emph{long enough arms}. We again reduce from 3-partition. Given a multiset of $3m$ positive integers $S = \{s_1,...,s_{3m}\}$, it is strongly NP-complete to decide if they can be partitioned into $m$ sets where all sets have the same sum \cite{Garey}. This result clearly still holds even when $m$ is even. So suppose $m$ is even and, let $\sum_{i=1}^{3m} s_i$ = M. Then $B = M/m$ is the required sum for each partition.

Let $l = \frac{7m^2}{2} - \frac{3m}{2} + 4$. Starting with a backbone $\beta=b_1...b_{4m+1}$, construct a comb in the following way: For each $s_i \in S$, add a corresponding arm of length $ls_i$ attached to $\beta$ at $b_{\frac{m}{2} + i}$. Add $m$ arms $c_1,...,c_m$ of length one to be used as checkpoints, with $c_i$ attached at $b_{\frac{m}{2}-\frac{i-1}{2}}$ if $i$ is odd, and $b_{\frac{7m}{2} + \frac{i}{2}}$ if $i$ is even. Add a single long arm of some length $k > 2lBm + \frac{7m^2}{2} - \frac{2m}{2} + 3$, attached at $b_{4m+1}$. For the TVG used in this proof (e.g., Figure \ref{bCombExample}), arms over any period of time are in one of three modes: $steady$, $flashing$, or $carrying$ (see the proof of Theorem \ref{bapx}).

Assume all edges in $\beta$ be available at all times, unless stated otherwise. Let $take$ be the temporal subgraph of duration $2lB + 3m$ in which the arms corresponding to $s_i$'s are steady, and all others are flashing. Let $check(j)$ be the temporal subgraph of duration $i + (i \bmod 2) + 2$ in which checkpoint $c_j$'s arm is steady, and all others are flashing. Let $finalcheck$ be the temporal subgraph of duration $\frac{m}{2} + 2$ in which checkpoint $c_m$'s arm is steady, and all others are flashing.  Let $finish$ be the temporal subgraph of duration $k$ in which $\beta$ is flashing, and all arms are flashing except for the long arm of length $k$, which is carrying. Let $\mathcal{G}$ be the TVG $take,check(1),take,check(2),...,take,check(m-1),take,finalcheck,finish$. The duration of $\mathcal{G}$ up until the start of $finish$ is $d = 2lBm + 3m^2 + \sum_{i=1}^m (i + (i \bmod 2) + 2) + \frac{m}{2} + 3 = 2lBm + \frac{7m^2}{2} - \frac{2m}{2} + 3$. The total duration of $\mathcal{G}$ is $D = d + k$.

Consider the problem of deciding whether DMVP over $\mathcal{G}$ with a single agent $a$ starting at $b_\frac{7m}{2}$ has a solution of length no more than $D$. Since $k > 2lBm + \frac{7m^2}{2} - \frac{2m}{2} + 3$, such a solution must take the long arm last, as traversing this arm twice would result in a solution of length greater than $2k > (2lBm + \frac{7m^2}{2} - \frac{2m}{2} + 3) + k = D$. 

Suppose there exists a 3-partition of $S$. During the $j$th $take$, if $j$ is odd (even), starting at $b_\frac{7m}{2}$ ($b_{\frac{m}{2}+1}$), $a$ can explore the arms of the spider corresponding to one partition, and end up at $b_{\frac{m}{2}+1}$ ($b_\frac{7m}{2}$) in exactly $2Bl + 3m$ steps. During the subsequent $check(j)$, $a$ visits $c_j$, and returns to $b_{\frac{m}{2}+1}$ ($b_\frac{7m}{2}$) in exactly $i + (i \bmod 2) + 2$ steps. During $finalcheck$, $a$ travels from $b_\frac{7m}{2}$ to $c_m$ and then to $b_{4m+1}$ in the allotted $\frac{m}{2} + 2$ steps. Finally, during $finish$, $a$ takes the long arm, reaching its leaf without waiting, completing coverage in $D$ steps. 

Now, suppose this instance of DMVP has a solution of length $D$. If there is no 3-partition of $S$, then $a$ must wait during traversal of at least one arm corresponding to an $s_i$. Since the length of this arm is $ls_i$, $a$ must in fact wait for at least $l$ edges of this arm, incurring a cost of $(\Delta-1)l = (\Delta-1)(\frac{7m^2}{2} - \frac{3m}{2} + 2)$. Since in the length $D$ solution described above $a$ never waits on an arm, this incurred cost must be made up for by minimizing traversal of edges in $\beta$. However, in the length $D$ solution described above, $a$ only traverses $\beta$ for a total of $\frac{7m^2}{2}-\frac{3m}{2}+1 < (\Delta-1)l$ steps, $\forall \Delta > 1$. Therefore, the solution must be in the form of the solution described above in which a 3-partition of $S$ does indeed exists. So, the decision problem is NP-complete, since 3-partition remains NP-complete even when input integers are given in unary.

Consider the minimization version of the problem with the same setup. Note that if $a$ does not begin taking the long arm right as $finish$ begins, the long arm will take at least $\Delta(k-1)+1$ to traverse, this best case occurring when $a$ does not have to wait for the first edge. In particular, if $a$ takes the long arm last, but has not visited all other arms before $finish$ starts, visiting those arms must have taken at least $d + (\Delta - 1)$, since $a$ must wait at least once during their traversal, and the total cost of the solution is then at least $D' = \Delta(k-1) + 1 + d + (\Delta - 1) = \Delta k-1 + d$. If $a$ does not take the long arm last, it must traverse the long arm twice, taking at least $\Delta(k-1) + 1 + k$ steps (this best case occurring when $a$ starts taking the long arm right as $finish$ starts), and once it returns must wait at least once while traversing the remaining arms, making the length of the total solution at least $D'' = \Delta(k-1) + 1 + k + d + (\Delta - 1) = \Delta (k+1) + d > D'$. Take any real constant $\delta < \Delta$. Choose the least integer $N$ s.t. $N > \frac{1}{\Delta-\delta}$. Let $k = N\delta d$. Then 
				$$(\Delta-\delta)N\delta d > \delta d,$$
				$$(\Delta-\delta)k > \delta d,$$
				$$ \Delta k > \delta d + \delta k,$$
				$$ \Delta k + d = D' > \delta D, \forall \Delta > 1.$$
Therefore, any solution that contains waiting cannot be a $\delta$-approximation. So, finding a $\delta$-approximation is equivalent to finding a solution with no waiting, i.e., a minimal solution, and thus is NP-hard. Hence, the problem is NP-hard to approximate within any factor less than $\Delta$. 
 \end{proof}
 
 \begin{figure}
\centering
\includegraphics[scale=1]{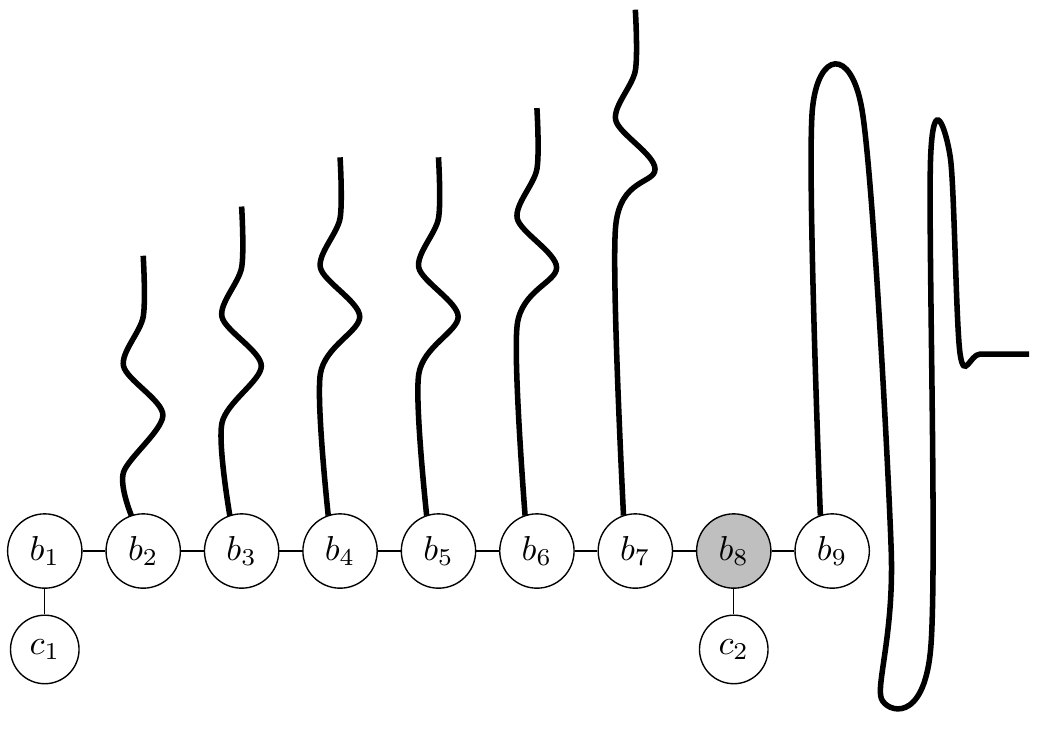}
\caption{3-partition underlying comb for some $S$ with $\lvert S \rvert = 6$. Agent starts at $b_8$. \label{bCombExample}}
\end{figure}

As is shown in Section~\ref{upperBounds}, there is a much greater potential for tractability of DMVP in $\mathcal{P}$ than in $\mathcal{B}$ or $\mathcal{R}$. However, the next result follows immediately via reduction from hamiltonian path by simply restricting $t$ to $n-1$.

\begin{theorem}
\label{phard}
DMVP for a single agent in $\mathcal{P}$ is NP-complete, when $p=1$.
\end{theorem}
\begin{proof}
$p=1$ is simply the static case, so the theorem follows immediately from the result that MVP is NP-complete for a single agent on general graphs \cite{MVP}.
\end{proof}

DMVP in $\mathcal{P}$ for $p=1$ is then also NP-complete for all classes of graphs for which hamiltonian path is NP-complete, in particular, planar graphs of maximum degree 3, bridgeless undirected planar 3-regular bipartite graphs, and 3-connected 3-regular bipartite graphs \cite{Saito80}. To show that $\mathcal{P}$ is an interesting dynamics class for DMVP in its own right, it is important to show that DMVP yields different hardness results over $\mathcal{P}$ than over static graphs. Thus, we construct a class of graphs for the following result:

\begin{theorem}
\label{p2}
There is an infinite class of graphs $C$ such that DMVP for a single agent in $\mathcal{P}$ over graphs in $C$ is NP-complete when $p=2$, but trivial when $p=1$.
\end{theorem}
\begin{proof}
Given a graph $G$ with an even number of vertices arbitrarily ordered $v_0,...,v_{n-1}$, construct a corresponding graph $G' \in C$ by adding $n$ new vertices $c_0,...,c_{n-1}$, and adding the edges $(v_i,c_i), (v_i,c_{i+1}),$ and $(c_i,c_{i+1})$ for all $0 < i < n$, where indices are taken mod $n$. 

To show the problem is NP-complete for a single agent in $\mathcal{P}$ over graphs in $C$, with $p=2$, we reduce from the hamiltonian path problem \cite{Karp}. Consider a graph $G$ with an even number of vertices $n$, and one of those vertices $v_0$, with the problem of deciding whether $G$ contains a hamiltonian path starting at $v_0$. Take the graph $G' \in C$ corresponding to $G$ as the underlying graph of $\mathcal{G}$. $\mathcal{G}$ begins at time 0. In $\mathcal{P}$ with $p=2$, traversable edges can only be one of three possible types: (01) available at odd times but not even times, (10) available at even times but not odd times, (11) available at all times. Let all original edges of $G$ be of type 11. Let $(v_i,c_i)$ be of type 01 when $i$ is even and type 10 when $i$ is odd. Let $(v_i,c_{i+1})$ be of type 10 when $i$ is even and 01 when $i$ is odd. Let $(c_i,c_{i+1})$ be of type 10 when $i$ is even and 01 when $i$ is odd (see Figure \ref{P2Example}). 

Consider the problem of deciding if DMVP over $\mathcal{G}$ for a single agent $a$ starting at $v_0$ has a solution of length no more than $2n-1$, i.e., a solution with no waiting, and in which each vertex is visited exactly once. This problem is clearly in NP. If there is a hamiltonian path in $G$ from $v_0$ vertex $v_i$, then this path will be constantly available in $\mathcal{G}$. $a$ can take this path in $n-1$ steps, and, ending at an odd time, immediately follow $(v_{i},c_{i})$ if $i$ is even, or $(v_{i},c_{i+1})$ if $i$ is odd, then follow either the path $c_ic_{i+1}c_{i+2}...c_{i-1}$ or $c_{i+1}c_{i+2}c_{i+2}...c_i$, the edges for which are always available as $a$ reaches the incoming vertices, thus completing the overall traversal in exactly $2n-1$ steps. Suppose there is a solution to this problem of length $2n-1$. By construction, if $a$ moves to any $c_i$ before covering every $v_j$, $a$ must then wait at least once at some $c_k$ before visiting any further $v_l$. This is because for all $c_i$, once $c_i$ is reached via either $(v_{i-1},c_i)$, $(v_i,c_i)$, or $(c_{i-1},c_i)$ the only edge that can be immediately taken without waiting is $(c_i,c_{i+1})$. So $a$ must visit all $v_j$ exactly once without visiting any $c_i$, thus following a path corresponding to a hamiltonian path in $G$.
	
However, if we consider the same setup over $G'$ but with $p = 1$, $v_0c_1v_1c_2v_2...$\\$c_{n-1}v_{n-1}c_0$ is always an optimal solution.
\end{proof}

\begin{figure}
\begin{center}
\includegraphics[scale=1]{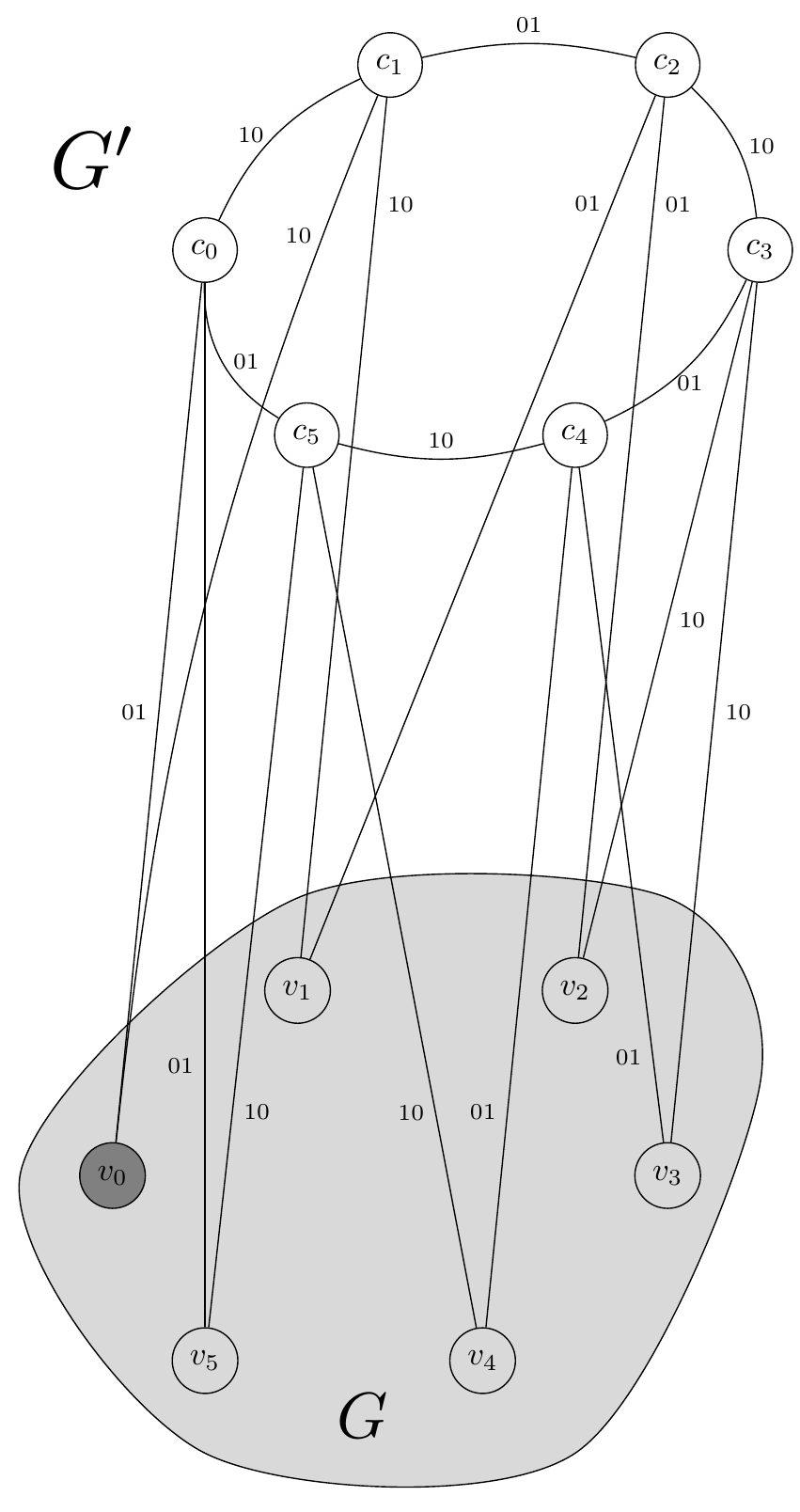}
\end{center}
\caption{$\mathcal{G}$ (from Thm. \ref{p2}) with underlying graph $G'$ constructed from some six-node graph $G$. Edges labeled 10 are available at even times; 01 at odd times. \label{P2Example}}
\end{figure}

\section{Upper Bounds}
\label{upperBounds}

In this section, we map out a class of graphs over which DMVP in $\mathcal{R}$ is solvable in polynomial time. We first start with a very useful lemma. Note that a related observation (about turning around on a ring) was made in \cite{Wade13}.

\begin{lemma}[Turning around lemma]
\label{turnAround}
There is always an optimal solution $J$ that never turns around at a degree 2 vertex of the edge-induced subgraph of $J$ in $G$.
\end{lemma}
\begin{proof}
Suppose $v$ is a degree 2 vertex with neighbors $u,w$ in the edge-induced subgraph of $J$ in $G$. Suppose agent $a$ takes edge $(u,v)$ at time $\tau$, then \emph{turns around}, taking $(u,v)$ at time $\tau'$ as the very next edge in its traversal. Since $(v,w)$ is in the edge-induced subgraph of $J$, $a$ must traverse $(v,w)$ at some other time, thereby reaching $v$ at that time. So, $a$ could have waited at $u$ from times $\tau$ to $\tau' + 1$, instead of taking $(u,v)$ at time $\tau$, and the solution would still be optimal (see Figure~\ref{waysToCover}). 
\end{proof}

\begin{figure}[t]
\begin{center}
\includegraphics[scale=1.8]{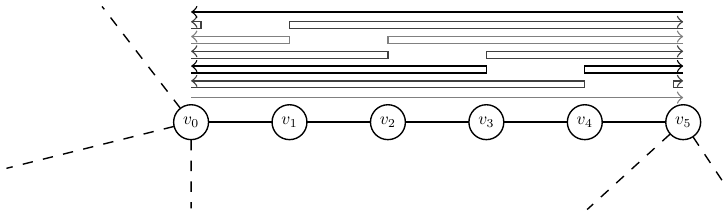}
\end{center}
\caption{The 7 ways, satisfying Lemma 1, of covering the vertices of a length 5 path with degree 2 intermediate nodes. \label{waysToCover}}
\end{figure}

We apply Lemma~\ref{turnAround} to get the following solvability results for restricted classes of underlying graphs.

\begin{theorem}
\label{line}
DMVP for a single agent in $\mathcal{R}$ on a path is solvable in $O(T)$ time.
\end{theorem}
\begin{proof}
Consider DMVP for a single agent $a$ with underlying graph $G$ the path $v_0v_1...v_n$, and $a$ starting at $v_k$. To reach $v_0$, $a$ must cover all $v_{k-1},...,v_1$ along the way. Similarly, to reach $v_n$, $a$ must cover all $v_{k+1},...,v_{n-1}$. By Lemma \ref{turnAround}, an optimal solution can be constructed by first taking a foremost journey to either $v_0$ or $v_n$, then taking the foremost journey to the remaining endpoint. One of these two topological journeys, called the $left$ and $right$ journeys, must embody an optimal solution, but in the worst case edge availability must be checked for all $t \in \mathcal{T}$, yielding an $O(T)$ runtime. 
\end{proof}
\begin{algorithm}[t]
\caption{DMVP-Line($\mathcal{G}, \{v_k\})$ \label{linealg}}
\begin{algorithmic}
\State $lLoc = rLoc = k$ \Comment{Both possible journeys start at $v_k$}
\State $lTurned = rTurned = complete = False$
\State $t = 0$
\While{not $complete$}
	\If{not $lTurned$} \Comment{Advance left journey if possible}
		\If{$\rho((v_{lLoc},v_{lLoc-1}),t) = 1$}
			\State $lLoc = lLoc - 1$
			\If{lLoc = 0} \Comment{Left endpoint reached}
				\State $lTurned = True$
			\EndIf
		\EndIf
	\Else
		\If{$\rho((v_{lLoc},v_{lLoc+1}),t) = 1$}
			\State $lLoc = lLoc + 1$
			\If{rLoc = n} \Comment{Right endpoint reached}
				\State $complete = True$
			\EndIf
		\EndIf
	\EndIf
	\If{not $rTurned$} \Comment{Advance right journey if possible}
		\If{$\rho((v_{rLoc},v_{rLoc+1}),t) = 1$}
			\State $rLoc = rLoc + 1$
			\If{rLoc = n} \Comment{Right endpoint reached}
				\State $rTurned = True$
			\EndIf
		\EndIf
	\Else
		\If{$\rho((v_{rLoc},v_{rLoc-1}),t) = 1$}
			\State $rLoc = rLoc - 1$
			\If{rLoc = 0} \Comment{Left endpoint reached}
				\State $complete = True$
			\EndIf
		\EndIf
	\EndIf
	\State $t = t + 1$ \Comment{Step}
\EndWhile
\State \Return $t$
\end{algorithmic}
\end{algorithm}
 
\begin{theorem}
\label{cycle}
DMVP for a single agent in $\mathcal{R}$ on a cycle is solvable in $O(Tn)$ time.
\end{theorem}
\begin{proof}
A similar case to Theorem \ref{line} can be made for the cycle $C = v_0v_1...v_nv_0$. Suppose $a$ starts at $v_0$ at time $0$. Consider an optimal visitation of $C$ for $a$. In this optimal solution, there is some vertex $v_k \neq v_0$ that is visited last. The second to last vertex is then either $v_{k-1}$ or $v_{k+1}$. If it is $v_{k-1}$, then $a$ must have already visited $v_{k+1}$ without visiting $v_k$. So, the edge $(v_{k+1},v_k)$ is never traversed. Therefore, the solution reduces to an optimal solution over the path graph $v_kv_{k-1}...v_{k+1}$ starting at $v_0$. Similarly, if instead $v_{k+1}$ is the vertex visited second to last, then $a$ must have already visited $v_{k-1}$ without visiting $v_k$, and the solution reduces to an optimal solution over the path $v_{k-1}v_{k-2}...v_{k+1}v_k$. Since there are $n-1$ possible final vertices for an optimal solution, the cost of an optimal solution can be computed by for each of these vertices computing the minimal cost between optimal coverage of each of the two corresponding paths using Algorithm~\ref{linealg}, and taking the minimum over all $n - 1$ vertices possible final vertices (see Figure \ref{waysToCycle}). This yields an $O(Tn)$ runtime. 
\end{proof}

\begin{figure}
\centering
\includegraphics[scale=2]{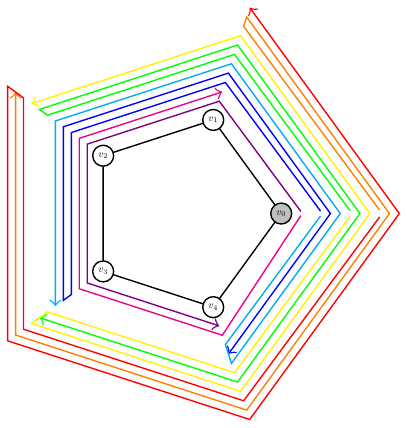}
\caption{The 8 possible underlying walks of solutions, satisfying Lemma 1, to the 5-cycle starting at $v_0$. \label{waysToCycle}}
\end{figure}

Now we show that despite the severe inapproximability of DMVP over $\mathcal{R}$, we can always compute an optimal solution in exponential time.

\begin{theorem}
\label{dynam}
DMVP for a single agent in $\mathcal{R}$ is solvable in $O(Tn^3+n^22^n)$ time.
\end{theorem}
\begin{proof}
The proposed algorithm first computes all-pairs-all-times-foremost-journey of input TVG $\mathcal{G}$, using a straightforward dynamic programming algorithm, then uses this information to run another dynamic programming algorithm, conceived along the lines of a standard method for TSP \cite{Bellman}. 

Let $d^t_{uv}$ be the length of the foremost journey from $u$ to $v$, starting at time $t$. Algorithm \ref{apatfj1} computes $d^t_{uv}$ for all vertex pairs $(u,v)$, and times $t \in \mathcal{T}$ for a given TVG $\mathcal{G}$.

At all times $t$, for all vertices $u \in V$, $d^t_{uu}$ is clearly 0. At time $T$, the time limit has been reached, so an agent cannot move to another vertex in any guaranteed time, and thus we set $d^T_{uv} = \infty$ for all $u \neq v$. For all $T - 1 \geq t \geq 0$, in the worst case an agent can wait at $u$ for one step, and take the foremost journey to $v$ starting at time $t+1$. If there is a better journey than this, it must consist of not waiting, rather taking one of the edges available at time $t$ from $u$ to some vertex $k$. Subsequently taking the foremost journey from $k$ to $v$ starting at time $t+1$ results in an optimal journey through $k$. Algorithm \ref{apatfj1} clearly runs in $O(Tn^3)$ time, and uses $O(Tn^2)$ space.

\begin{algorithm}[h]
\caption{All-pairs-all-times-foremost-journey($\mathcal{G}$) \label{apatfj1}}
\begin{algorithmic}
\ForAll{$u,v \in V \times V$}	\Comment{Initialize base case for $t = T$.}
	\If{$u = v$}
		\State $d^T_{uv} = 0$ 
	\Else
		\State $d^T_{uv} = \infty$ \Comment{Since input ends at $T$, agent cannot move.}
	\EndIf
\EndFor	
\\
\For{$t = T-1,...,0$} \Comment{Work backwards until start time $t = 0$.}
	\ForAll{$u,v \in V \times V$}
		\If{$u = v$}
			\State $d^t_{uv} = 0$
		\Else
			\State $d^t_{uv} = d^{t+1}_{uv} + 1$ \Comment{In worst case, just wait at $u$.}
			\ForAll{$k \in V$}
				\If {$\rho((u,k),t) = 1$} \Comment{Check for better route.}
					\State $d^t_{uv} = \min(d^t_{uv},d^{t+1}_{kv} + 1)$ 
				\EndIf
			\EndFor
		\EndIf
	\EndFor
\EndFor

\end{algorithmic}
\end{algorithm}

Algorithm \ref{expo1} uses the $d^t_{uv}$ values computed by Algorithm \ref{apatfj1} to compute the cost of a minimal solution to DMVP for a single agent in $\mathcal{R}$. Let $V' \subseteq V$ and $c(V',v)$ be the minimal time it takes to visit all vertices in $V'$ starting at vertex $s$ at time $0$ and ending at vertex $v \in V'$.

After initializing the minimal costs for visiting subsets up to size 2, the algorithm repeatedly uses the minimal costs for size $i$ subsets to calculate $c(V',v)$ for each size $i+1$ subset $V'$ and $v \neq s \in V'$. Once computed up to size $n$, the algorithm returns the minimal cost among journeys that cover all vertices. This is an optimal solution to DMVP as it is the minimum cost of taking foremost journeys between vertices that results in a complete cover. There are $2^n$ subsets of $V$, and so $n2^n$ subset-vertex pairs of the form $(V',v)$. For each of these, the algorithm computes the minimum of $O(n)$ values. So, Algorithm \ref{expo1} has running time $O(n^22^n)$. Since it saves one cost for each subset-vertex pair, Algorithm \ref{expo1} also uses $O(n2^n)$ space. Sequentially running Algorithm \ref{apatfj1} followed by Algorithm \ref{expo1}, we have a complete algorithm for DMVP for a single agent in $\mathcal{R}$, with combined running time $O(Tn^3+n^22^n)$. 
 \end{proof}

\begin{algorithm}[h]
\caption{$\mbox{\emph{DMVP}}(\mathcal{G},\{s\})$ \label{expo1}}
\begin{algorithmic}
\State $c(\{s\},s) = 0$ \Comment{Initialize subset of size 1.}
\ForAll{$v \neq s \in V$} \Comment{Initialize subsets of size 2.}
	\State $c(\{s,v\},v) = d^0_{sv}$
\EndFor
\For{i = 3,...,n} \Comment{Build up to subsets of size n.}
	\ForAll{$S \subseteq V s.t. \lvert S \rvert = i$}
		\ForAll{$v \neq s \in V$}
			\State $c(V',v) = \min_{u \neq s \in V' \setminus \{v\}}(c(V' \setminus \{v\}, u) + d^{c(V' \setminus \{v\},u)}_{uv})$
		\EndFor
	\EndFor
\EndFor
\Return $\min_{v \neq s \in V}(c(V,v))$
\end{algorithmic}
\end{algorithm}

Almost-trees have been previously studied with respect to fixed parameter tractability (e.g., \cite{Fiala01}). We use Theorem \ref{dynam} to generalize Theorems \ref{line} and \ref{cycle} with the following:

\begin{theorem}
\label{fpt}
DMVP in $\mathcal{R}$ is fixed parameter tractable, when $G$ is an $m$-leaf $c$-almost-tree, for fixed parameter $m$, and $c$ constant.
\end{theorem}
\begin{proof}
First, consider the restricted case where $G$ is an $m$-leaf tree. Since every leaf must be visited, and visiting all leaves implies coverage of the entire tree, there is a minimal solution that can be thought of as an ordering of the set of leaves of $G$, and the foremost journeys between them. In this case, there is only one \emph{way} to visit any node, namely, on the way to a leaf. Using this observation and Algorithm \ref{expo1} from the proof of Theorem \ref{dynam}, we see that we only need to consider all orderings of leaves, instead of all orderings of vertices, yielding a run time of $O(Tn^3 + m^22^m)$, which is indeed fixed parameter tractable for parameter $m$.

Suppose the underlying graph $G$ of $\mathcal{G}$ is an $m$-leaf $c$-almost-tree. Consider all edges $e$ such that removing $e$ from $G$ results in a $(c-1)$-almost-tree. Each of these edges lies on some path $P$ such that removing any edge of $P$ will similarly result in a $(c-1)$-almost-tree, and every intermediate vertex on the path has degree 2. Suppose $P$ is the path $v_0...v_l$. Since $G$ is an $m$-leaf $c$-almost-tree, there are $O(m)$ paths of this type. The edge-induced subgraph $G'$ of the underlying walk of an optimal covering of $\mathcal{G}$ can be any $(c-c')$-almost-tree $\subseteq G$, for $0 \leq c' \leq c$. For each $c'$, a solution involves selecting $c'$ paths, each of $O(n)$ length, from which to remove an edge. So, there are $O(m^{c'}n^{c'})$ possible choices of $(c-c')$-almost-trees, and thus $O(\sum_{c'=0}^c (m^{c'}n^{c'})) = O(m^cn^c)$ choices for $G'$. Every $G'$ has no more than $m+2c$ leaves. Since every edge of $G'$ is covered, by Lemma~\ref{turnAround}, there are at most 2 ways to cover each of the remaining $O(m)$ paths  $v_0...v_l$ of intermediate vertex degree 2, namely: entering at $v_0$ and exiting at $v_l$, or entering at $v_l$ and exiting at $v_0$. Augment the set of leaves to be ordered in a solution with the selected ways of covering these paths, that is, select one of the consecutive subsequences $v_0v_l$ or $v_lv_0$ to be in the ordering. With this augmentation, we still have a set of $O(m)$ elements to be ordered, the optimal ordering of which can be computed via Theorem \ref{dynam} in $O(Tn^3 +  m^22^m)$ time. The minimum over all ways of covering $G'$ can then be computed in $O(2^m)O(Tn^3 +  m^22^m) = O(Tn^32^m + m^22^{2m})$. The overall minimum cost for covering $\mathcal{G}$ can then be computed by taking the minimum cost over all $O(m^c n^c)$ edge-induced subgraphs in $O(m^{c'}n^c)O(Tn^32^m + m^22^{2m}) = O(Tn^{3+c}f(m))$ time.
 \end{proof}

The following result follows immediately for the case when $m = O(\lg n)$. 

\begin{corollary}
\label{poly}
DMVP in $\mathcal{R}$ is solvable in polynomial time, if $\mathcal{G}$ is an $O(\lg n)$-leaf $c$-almost-tree, for $c$ constant.
\end{corollary}

We conjecture (see Section \ref{openProbs}) that the maximal class of graphs over which DMVP in $\mathcal{R}$ is poly-time solvable is the class of all graphs with polynomially many spanning trees, all of which have $O(\lg n)$ leaves. 

Since DMVP in $\mathcal{B}$ is bounded by $2\Delta n$, the running time of the algorithm in Theorem \ref{dynam} on TVGs over $\mathcal{B}$ reduces to $O(\Delta n^4+n^22^n)$. We also see that we are able to greatly improve on approximation from $\mathcal{R}$ to $\mathcal{B}$:

\begin{theorem}
\label{btree}
DMVP in $\mathcal{B}$ over a tree can be $\Delta$-approximated in $O(n)$ time. This approximation is tight.
\end{theorem}
\begin{proof}
In \cite{MVP}, it is shown that minimal MVP cost $C$ can be computed in $O(n)$ for static graphs. In the dynamic case, the journey corresponding to following exactly the edges in the static solution when they are available can be followed, waiting at most $\Delta-1$ steps for each edge to appear before it is traversed. Since no solution can be better than $C$, and the proposed journey takes at most $\Delta C$, it must be a $\Delta$-approximation. From Theorems \ref{bapx} and \ref{bcomb}, we know there can be no better approximation. Hence, this approximation is tight.
\end{proof}

\begin{theorem}
\label{bapprox}
DMVP in $\mathcal{B}$ can be $2\Delta$-approximated by any online spanning tree traversal of $G$.
\end{theorem}
\begin{proof}
The topological length of a spanning tree traversal is no more than $2n-3$. In the worst case, waiting $\Delta-1$ time steps for each subsequent edge to appear results in coverage of $\mathcal{G}$ in $2n\Delta - 3\Delta$ steps. The fastest possible coverage of $\mathcal{G}$ is via the traversal of a hamiltonian path in $G$ without waiting, which takes $n-1$ steps, and $2\Delta(n-1) > 2n\Delta - 3\Delta$. 
\end{proof}

Theorems \ref{bapx} and \ref{bcomb} show the tightness of Theorem \ref{btree}. Here, $\mathcal{B}$ is starkly differentiated from $\mathcal{R}$ in that we have at least some ability to approximate in $\mathcal{B}$. See Section \ref{openProbs} for a further discussion of the relationship between these two classes.

Similar to the case for $\mathcal{B}$, DMVP in $\mathcal{P}$ is bounded by $2pn$, so the running time of the algorithm in Theorem \ref{dynam} reduces to $O(pn^4+n^22^n)$. To exemplify the differences between $\mathcal{P}$ and $\mathcal{B}$, and motivate interest in the tractability of DMVP over $\mathcal{P}$, we first give the following simple example:

\begin{theorem}
\label{pcomb}
For any $p$, there is a class of problems over combs, for which DMVP in $\mathcal{B}$ is NP-hard, but in $\mathcal{P}$ is solvable by the online algorithm: take arms when you get to them.
\end{theorem}
\begin{proof}
Suppose an agent $a$ starts at $b_0$. $a$ can either take $A_0$ immediately, or travel along $B$ to visit other arms and return to visit $A_0$ at a later time. Suppose the fastest an agent starting at $b_0$ can visit the leaf of $A_0$ and return to $b_0$ is $l$ steps. Then the longest this could possibly take $a$ starting at time 0 is $l + (p-1)$ steps, this worst case occurring when $a$ must wait $p-1$ steps for the fastest journey to become available. Suppose the fastest journey from $b_0$ to $b_k$ takes $k'$ steps. Then in the worst case, traveling from $b_0$ to $b_k$ takes $k' + (p-1)$ steps. Suppose the fastest coverage of the remaining induced subgraph $G' = G \setminus (A_0 \cup \{b_0,...,b_{k-1}\})$ takes $m$ steps. 

If $G'$ has only one arm, the foremost journey from $b_k$ to the leaf of this arm is clearly optimal. Assume that if $G'$ has $\alpha$ arms, then the following online algorithm results in an optimal solution: if at the endpoint of an unvisited arm, take that arm, otherwise visit the next unvisited vertex of $B$. In the $\alpha+1$ arm case, our agent starting at $b_0$ using this algorithm will complete coverage in no more than $l + (p-1) + k' + (p-1) + m = l + k + (2p-2) + m$ steps. But any solution in which $a$ does not take $A_0$ first must cost at least $l + k' + k + m$, as $a$ must retraverse $b_k...b_0$ on its way back to cover $A_0$. Since $k \geq 2p-2$, the online solution must be minimal.

It is straightforward to modify Theorem \ref{bcomb} (i.e., by appropriately scaling up the underlying graph: elongating arms, extending the backbone, separating arms along the backbone, and adding an additional check tooth at $b_0$) to show that DMVP over the above class of combs is NP-hard in $\mathcal{B}$, for a single agent starting at $b_0$.
\end{proof}

\begin{figure}
\centering
\includegraphics[scale=1]{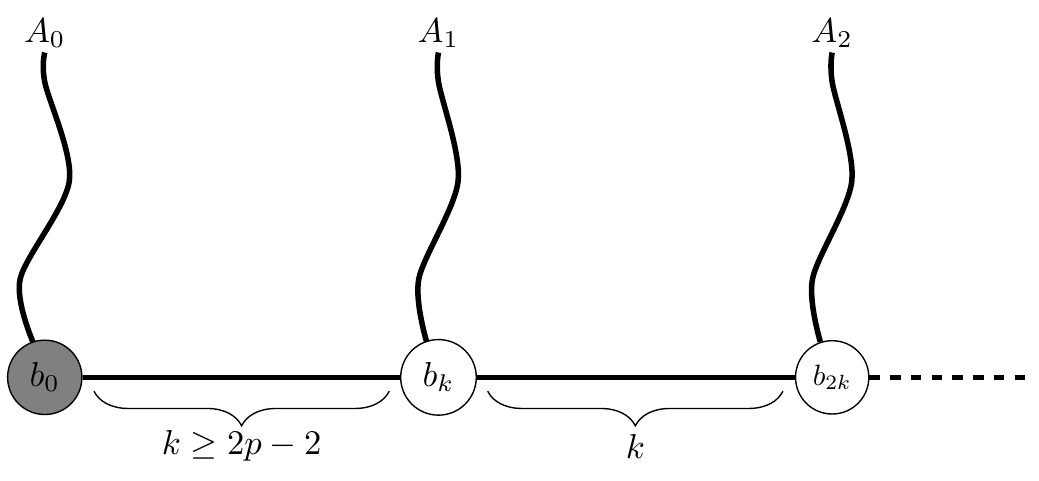}
\caption{Segment of an underlying comb for which DMVP  in $\mathcal{P}$ for an agent starting at $b_0$ is solved by the simple online tree traversal algorithm, regardless of the lengths of each arm $A_i$ \label{imgLongComb}}
\end{figure}

The quality of $\mathcal{P}$ we take advantage of above is that if the fastest journey between two nodes takes $d$ steps, the foremost journey can take no longer than $d + (p-1)$, while in $\mathcal{B}$ it can be as bad as $d\Delta$. We again harness this effect in the following result, a stronger theorem in the context of our inapproximability results for $\mathcal{R}$ and $\mathcal{B}$ (Theorems \ref{npostar} and \ref{bapx}):

\begin{theorem}
\label{pspider}
DMVP in $\mathcal{P}$ over a spider is solvable in polynomial time, for fixed $p$.
\end{theorem}
\begin{proof}
Starting at the center $c$ of the spider, it is never useful for an agent to travel along any arm, unless it reaches a leaf. That is, an optimal solution is essentially an optimal visitation order of the leaves. We can set up a cost function $c(l,t)$ giving the minimal time it takes to travel from $c$ to leaf $l$ and back, starting at time $t$. Since $\mathcal{G}$ is periodic, $c(l,t) = c(l,t+kp) \ \forall k \in \mathbb{Z}$. Suppose the fastest journey from $c$ to $l$ and back has cost $m(l)$. Let \emph{extra time} $e(l,t) = c(l,t) - m(l)$, be the cost above minimum incurred by traveling to $l$ and back starting at time $t$. $0 \leq e(l,t) < p \ \forall l, t$, since $a$ can always simply wait at most $p-1$ steps for the periodically fastest journey to be available. For any $l$, there are only $p^2$ possible functions $e(l,t)$, since for all $0 \leq t\pmod{p} \leq p-1$, $0 \leq e(l,t) < p$. Let $r(l,t)$ be the \emph{return time} $\mod p$ of traveling to $l$ and back, that is, $c(l,t) = i \implies r(l,t) = t + i \pmod{p}$. Classify each $l$ by $e$ and $r$. Let $l_1 \equiv l_2 \iff e(l_1,t) = e(l_2,t)$ and $r(l_1,t) = r(l_2,t) \ \forall t$. Since for each $t$ there are $p$ choices for $e$ and $p$ choices for $r$, there are $p^3$ such equivalence classes. During a traversal of the spider, leaves with a common equivalence class are interchangeable, since at a given time, taking any of the same class will result in equivalent incurrence of cost above minimum as well as equivalent return time. Thus, a minimal traversal consists of traversing an ordering of arms corresponding to a sequence of equivalence classes $q_i$, such that the frequency of $q_i$ in the sequence is the number of arms classified as $q_i$. 

Notice that for any length $p$ traversed sequence $q_1q_2...q_p$, by the pigeonhole principle, there must be $q_i, q_j$ with $i < j$ such that $r(q_i,t_i) = r(q_j,t_j)$, where $t_i$ is the start time for traversing the $q_i$ arm, and $t_j$ is the start time for traversing the $q_j$ arm. Let $Q^t$ be a \emph{pattern} if $Q$ is a sequence of equivalence classes with $0 < \lvert Q \rvert \leq p$, and starting at $c$ at time $t$, the traversal of $Q$ returns at some time $t' \equiv t \pmod{p}$. Furthermore, $Q^t$ is not a pattern if it contains any subpatterns, i.e., traversed subsequence with equivalent start and end time. Any length $p$ sequence must contain a pattern. An optimal solution can be decomposed into a sequence alternating between patterns and non-pattern subsequences between patterns. Any pattern can be removed from its location starting at time $t$ and inserted at any different location $t' \equiv t \pmod{p}$, since the fact that the pattern has the equivalent start and end time means adjacent journeys will be unaltered, due to the periodicity of $\mathcal{G}$. In particular, any pattern $Q_1^t$ can be removed from its current location and inserted after any $Q_2^t$, without changing the cost of the solution. So, given any optimal solution, the following reordering process does not change the cost of the solution:

\begin{enumerate}
\item Divide the sequence into patterns $Q^t$ and stray arms in classes $q$ not in patterns
\item Set $i = 0$, $S = \{0\}$
\item Sequence all $Q^i$ together starting at $i$
\item Identify new patterns created by this move
\item Repeat 2 and 3 until nothing changes
\item Consider earliest start time $j$ of an arm such that $j \pmod{p} \notin S$ after final $Q_i$
\item Let $i = j$, and add $j$ to $S$.
\item Repeat 3 through 8 until nothing changes
\end{enumerate}

After this process, all $Q_i$ are grouped together for all $i$, with fewer than $p$ stray arms separating each of these sequences of patterns, since in step 6, if there is no such $j$, then there must be fewer than $p$ arms left, otherwise there would be a pattern among these arms. Thus, the reordered sequence also ends with fewer than $p$ stray arms. Since we started with an optimal solution, and the above process did not change the cost of the solution, there must be an optimal solution of this form. Since every such reordering begins with the $Q_0$ patterns, there are $(p-1)!$ orders in which the $p$ grouped sequences of patterns can show up. For each of these, there are at most $p$ clusters of stray arms each of length less than $p$. There are $O(p^5)$ ways to fill up these $O(p^2)$ slots with arms from the $p^3$ classes. Since there are $O((p^3)^p)$ possible patterns, there are $O(n^{(p^3)^p})$ ways to partition the remaining $O(n)$ arms in each class into patterns, and therefore $O(n^{(p^3)^pp^3}) = O(n^{(p^3)^{p+1}})$ ways to partition all classes into patterns. This yields $O((p-1)!p^5n^{(p^3)^{p+1}}) = O(n^{(p^3)^{p+1}})$ possible solutions of the form reached by executing steps 1-8 above, at least one of which must be optimal. The cost of each of these possible solutions, of which there are polynomially many in $\lvert \mathcal{G} \rvert$, can be easily computed in time polynomial in $\lvert \mathcal{G} \rvert$. If $a$ does not start at $c$, but rather on some arm $A$, $a$ can either visit the leaf of $A$ before returning or $c$, or return to $c$ directly. Compute DMVP for the remaining arms in each of these two cases will yield an overall optimal solution still in time polynomial in $\lvert \mathcal{G} \rvert$. 
 \end{proof}

This polynomial runtime can be significantly improved for the case of $p=2$.

\begin{theorem}
\label{p2GeneralTree}
DMVP in $\mathcal{P}$ over a tree is solvable in $O(n)$ time, when $p = 2$.
\end{theorem}
\begin{proof}
Consider DMVP in $\mathcal{P}$, with $p = 2$, over a tree $T$, for an agent starting at root $o$ at time 0. The proof proceeds as follows: We first show by induction that there is always an optimal solution that never enters any of the subtrees of $o$'s children more than once. We then show that when covered in its entirety, each subtree is of one of three types: ($fw11$) fastest coverage with return to root is always available, ($fw10$) fastest coverage is available only at even times, and ($fw01$) fastest coverage is available only at odd times. Alternating between $fw10$ and $fw01$ subtrees, and then taking the remaining subtrees in any order, before ending at a \emph{furthest leaf} results in an optimal solution, as we maximize how many subtrees are traversed optimally. We can recursively compute the type and costs of covering the maximal subtree rooted at each node $v$, in $O(\deg(v))$ time for each.

Suppose $o$ has adjacent edges $e_1 = (o,u_1),...,e_d = (o,u_d)$. Let  $T^{u_i}$ be the maximal subtree rooted at $u_i$. Suppose agent $a$ starts at $o$. Recall that since $T$ is a tree, an optimal solution can be characterized as the set of leaves ordered by when they are visited.

Suppose an optimal solution ends on some leaf of $T^{u_i}$. Assume that when $T^{u_i}$ can be entered at most $k$ times during a solution, the solution that enters only once is still optimal.  Suppose an optimal solution $S_{k+1}$ enters $T^{u_i}$ $k+1$ times. Then there is some non-empty subgraph $T^{u_i}_k$ of $T^{u_i}$ that $a$ covers the upon the $k$th entry, and $T^{u_i}_{k+1}$ it covers upon the $(k+1)$st entry. Let $T'$ be the subgraph of $T$ covered between the $k$th and $(k+1)$st entries into $T^{u_i}$. When $a$ arrives at $o$ before entering $T^{u_i}$ for the $k$th time, consider an alternate completion resulting in an alternate solution $S_k$, in which $a$ instead immediately covers $T'$ and ends up back at $o$, at a cost of at most 1 plus the cost of this coverage in $S_{k+1}$ (which must occur in $S_{k+1}$ after the $kth$ entry into $T^{u_i}$). If $S_k$ does incur a cost over $S_{k+1}$ for this traversal, then both $S_k$ and $S_{k+1}$ reach $o$ at times $\tau$ and $\tau'$ with equal parity. In $S_k$, $a$ completes coverage of $T$ during its $k$th entry to $T^{u_i}$, by covering $T^{u_i}_k$ and returning to $v_i$, then immediately covering $T^{u_i}_{k+1}$. If $\tau$ and $\tau'$ have equal parity, then $T^{u_i}_{k+1}$ is traversed in both $S_k$ and $S_{k+1}$ in equal time. Otherwise, $S_k$ may incur a cost of 1 over $S_{k+1}$ for this traversal. Therefore, the combined coverage of $T'$ and $T^{u_i}_k$ costs at most 1 more in $S_k$ than in $S_{k+1}$. $S_k$ may incur an additional cost of 1 over $S_{k+1}$ for its completing coverage of $T^{u_i}_{k+1}$, for a total incurrence of at most 2 over $S_{k+1}$ for these traversals. However, $S_{k+1}$ contains two additional traversals of $e_i$, so the total cost of $S_{k+1}$ can be no better than that of $S_k$. Therefore, the solution that enters $T^{u_i}$ only once is optimal.

Recall (from the proof of Theorem \ref{p2}) that in $\mathcal{P}$ with $p=2$, there are only three relevant dynamic edge types: 01, which are available at odd times but not even; 10, which are available at even times but not odd; and 11, which are available at all times. Let $T^u$ be a maximal subtree of $T$ rooted at $u$. Let $c_w(T^u,i)$ be the cost of the foremost journey covering $T^u$, starting at time $t \equiv i \bmod 2$, with returning to end back at $u$. Let $c(T^u,i)$ be the  cost of the foremost journey covering without necessarily returning to $u$. Let $m_w(T^u)$ be the fastest cost of covering $T^u$ with returning to end back at $u$, and $m(T^u)$ be the fastest cost of covering $T^u$ without necessarily returning to $u$. Notice that, since $p = 2$, the foremost cost of these coverings at any time can be no more than $c_w(T^u,i) = m_w(T^u) + 1$ and $c(T^u,i) = m(T^u) + 1$, respectively. Classify $T^u$ as $fw11$ if a fastest coverage with return is always available; $fw10$ if it is available at even times, but not odd; $fw01$ if at odd times, but not even. 

To enable our characterization of an optimal ordering of subtrees, it is necessary that for any tree $T$, if fastest coverage of $T$ with return to $o$ is always available, then fastest coverage takes even time, otherwise, fastest coverage takes odd time. To see this, first consider height 1 trees (i.e., stars). Regardless of when $a$ starts, if the numbers of 01 and 10 edges are equal, then $a$ can alternate between them, taking each at cost 3, and taking all 11 edges at cost 2, the total cost of which will be even. If $a$ starts at time 0 and there are more 10's than 01's, then $a$ can take one more 10 than 01 at cost 3, and take the rest at cost 4, for an odd total. Starting at time 1 $a$ cannot save on this extra edge. A similar case applies when there are more 01's than 10's, but for opposite start times. For trees of greater depth, since each subtree need not be entered more than once, we can use similar reasoning to the height 1 case. If there are the same number of $fw01$ as $fw10$ subtrees, alternate between them to take them all fastest (each at odd cost), so fastest coverage is always available at even cost. Otherwise, one more subtree of the more plentiful type could be taken fastest, depending on start time.   

Since each maximal subtree is covered independently, we can (as in the proof of Theorem \ref{pspider}) look for \emph{patterns} in a solution given as an optimal ordering of child subtrees. Notice that in this case, as a result of the parity of coverage with return to $o$, there are only three patterns: $(fw11)$, $(fw10,fw01)$, and $(fw01,fw10)$. Given the above information for all of a tree $T^v$'s maximal child subtrees, we compute these values for $T^v$: 

Following the characterization of an optimal solution given in the proof of Theorem \ref{pspider}, we can construct $S_0$, an optimal solution with return that starts at time $t \equiv 0 \bmod 2$ by first taking as many copies of $(fw10,fw01)$ as possible, since each of these will result in fastest traversals of the covered subtrees (because fastest coverage of $fw10$ and $fw01$ arms always takes odd time), taking one more $fw10$ if possible, before taking the remaining subtrees in any order. We construct $S_1$, a similar solution with return that starts at time $t \equiv 1 \bmod 2$, by taking an $fw01$ subtree before constructing a time 0 solution in the same way as $S_0$. Since we know the form of an optimal solution, we can easily compute the cost of each of these two solutions in $O(\deg(v))$. Then, for each $i \in \{0,1\}$ and each subtree $T^u$ of $T^v$, we calculate the cost of covering $T_v$ without return starting at time $t \equiv i \bmod 2$, such that $T^u$ is the final subtree covered. We can do this by subtracting the cost of covering $T^u$ in $S_i$, adding 1 if the removal of $T^u$ from $S_i$ necessarily decreases the number  of subtrees taken optimally in $S_i$, and adding the cost of taking $T^u$ without return at the end of this new solution. The foremost cost of covering $T^v$ is the cost of the minimum solution over all $T^u$. Given the classification of $T^u$, this computation takes constant time for each $T^u$, and thus $O(\deg(v))$ overall. Given these costs, it is trivial to classify $T^v$. So, we can recursively compute all required values, at a cost of $O(\deg(v)$ per node, and thus $O(n)$ overall. The optimal cost of a complete solution is then $c(T,0)$.
 \end{proof}

We hypothesize that more efficient algorithms, such as the one for the $p=2$ case, exist for this type of problem for greater values of $p$, and even general $p$, via this method of piecing together fast patterns. We have similar high hopes for larger classes of underlying graphs.

\section{Open Problems and Discussion}
\label{openProbs}

This paper presents significant advances towards isolating the maximal class of graphs over which DMVP in $\mathcal{R}$ is solvable in polynomial time. We conjecture that this maximal class is the class of all graphs with polynomially many spanning trees, all of which have $O(\lg n)$ leaves. Furthermore, we conjecture that this class is equivalent for $\mathcal{R}$ and $\mathcal{B}$. But we are very interested in expanding this class with respect to $\mathcal{P}$, motivated by our solvability results for $\mathcal{P}$ over subclasses of trees. We have shown that for the case of $p=2$, DMVP for a single agent over general trees can be computed in linear time. This result relies on the fact that we know how to optimally piece together patterns with period 2. New methods for finding optimal pattern sequences could greatly reduce computation for cases of $p>2$. We are hopeful that DMVP in $\mathcal{P}$ will be shown to be poly-time solvable over arbitrary trees or at least bounded degree trees, for greater values $p$ both fixed and not fixed.

Considering $\mathcal{B}$ and $\mathcal{R}$, $\mathcal{B}$ is clearly differentiated from $\mathcal{R}$ in that we have at least some ability to approximate in $\mathcal{B}$. There remains, however, an important open question: Is there any class $C$ of underlying graphs such that DMVP is NP-hard over $C$ in $\mathcal{R}$, but not in $\mathcal{B}$? We are particularly interested in whether or not DMVP in $\mathcal{B}$ is NP-hard when the underlying graph is a star and $\Delta$ is fixed, in particular, when $\Delta = 2$. Note: The proof of Theorem \ref{npostar} implies it is hard when $\Delta$ is some relatively small function of the input. We conjecture that even for $\Delta = 2$ this problem is NP-hard, but the highly-restricted nature of the input makes an answer to this problem more elusive than some of the others we have results for. Towards an answer to this question, we give the following observation:

\begin{obs}
\label{noWaiting}
DMVP
in $\mathcal{R}$ over a spider with arms of uniform length $l$, e.g., a star (when $l = 1$), can be decided in polynomial time, when $t$ disallows waiting, i.e., $t = 2n - l - d$, where $d$ is topological distance from $s$ to $c$.
\end{obs}
\begin{proof} 
Suppose $G$ is a spider with arms of uniform length $l$. Then, $G$ has $n/l$ arms. Suppose $a$ starts at some vertex $s$ distance $d$ from the central vertex $c$. If $t = 2n - l - d$, then the only solution can be a waiting-free spanning tree traversal of $G$ starting at $s$. If $d > 0$, i.e., $s \neq c$, then the first leaf visited must be the leaf of $s$'s arm. If $a$ starts at $c$, any arm can be traversed first. In either case, starting at the first time $\tau$ that $a$ finds itself at $c$, $a$ must traverse the $O(n/l) = \alpha$ remaining arms $a_1,...,a_\alpha$ each in time ${2l}$, except for the final traversed arm, whose leaf is reached in $l$ steps from $c$, completing the solution. Starting at $\tau$, break the remaining time into $\alpha-1$ length $2l$ time blocks $b_1,...,b_{\alpha-1}$, and a final length $l$ time block $b_\alpha$. For each $a_i$, for each $b_j$, we can straightforwardly compute, in $O(l)$ time, whether or not $a$ can traverse $a_j$ and return to $c$ during $b_j$ without waiting. Deciding whether or not there exists a complete traversal of all $\alpha$ arms without waiting then reduces to the problem of finding a perfect bipartite matching between arms and time blocks, for which there are many known efficient polynomial time algorithms, e.g., \cite{Hopcroft}.
\end{proof}

Overall, our results show some instances where DMVP is tractable as well as showing that DMVP faces difficult computational challenges for some natural classes of underlying topologies and dynamics. These challenges motivate research into online, multi-agent solutions to the problem, since in many cases having a complete global view of the present and future does not appear to be very helpful; moreover, in agent-oriented applications ranging from software agents to mobile robots, the information available to teams of agents can be bounded both temporally and geographically, and such online, multi-agent approaches could be well suited to agent dynamics without diminishing tractability. We have begun to take steps in this direction using edge markovian TVG models \cite{Baumann11}. In these types of stochastic environments, investigating interactive agent policies is an especially interesting direction to pursue.


\begin{thebibliography}{99}

\bibitem{MVP}
Aaron, E., Kranakis, E., Krizanc, D.: 
On the complexity of the multi-robot, multi-depot map visitation problem. 
IEEE MASS, 795--800 (2011)

\bibitem{Ahr06}
Ahr, D., Reinhelt, G.: 
A tabu search algorithm for the min-max k-chinese postman problem.
Comp. and Ops. Res., 3403--3422 (2006)

\bibitem{Saito80}
Akiyama, T., Nishizeki T., Saito N.: 
NP-completeness of the Hamiltonian cycle problem for bipartite graphs. 
Journal of Info. Proc. 3.2, 73--76 (1980)

\bibitem{Baumann11}
Baumann, H., Crescenzi, P., Fraigniaud, P.: 
Parsimonious flooding in dynamic graphs.
Distr. Comp. 24.1, 31--44 (2011)

\bibitem{Bellman}
Bellman, R.: Dynamic programming treatment of the travelling salesman problem.
JACM, 9.1, 61--63 (1962)

\bibitem{Blum94}
Blum, A., Chalasani, P., Coppersmith, D., Pulleyblank, B., Raghavan, P., Sudan, M.: The minimum latency problem:
Proc. of 26th STOC, 163--171 (1994)

\bibitem{Santoro10}
Casteigts, A., Flocchini, P.,  Quattrociocchi, W., Santoro, N.: 
Deterministic computations in time-varying graphs.
IFIP TCS, 111--124 (2010)

\bibitem{Santoro12}
Casteigts, A., Flocchini, P.,  Quattrociocchi, W., Santoro, N.: 
Time-varying graphs and dynamic networks.
IJPED 27.5, 387--408 (2012)

\bibitem{Lotker08}
Chen, A., Koucky, M., Lotker, Z.: 
How to explore a fast-changing world. 
Autom., Lang. and Prog. Springer, 121--132 (2008)

\bibitem{Choset01}
Choset, H.: Coverage for robotics: a survey of recent results.
Annals of Math and AI, 31, 113--126 (2001)

\bibitem{Correll08}
Correll, N., Rutishauser, S., Martinoli, A.: 
Comparing Coordination Schemes for Miniature Robotic Swarms.
Springer Tracts in Adv. Robo., 39, 471--480 (2008)

\bibitem{Easton05}
Easton, K., Burdick, J.: A coverage algorithm for multi-robot boundary inspection.
Proc. of ICRA, 727--734 (2005)

\bibitem{Edmonds73}
Edmonds J., Johnson, E.: Matching, euler tours and the chinese postman problem.
Mathematical Programming 5, 88--124 (1973)

\bibitem{Rao07}
Fakcharoenphol, J., Harrelson, C., Rao, S.: 
The k-traveling repairman problem. 
Proc. of 39th STOC (2007)

\bibitem{Fiala01}
Fiala, J., Kloks, T., Kratochvil, J.: Fixed-parameter complexity of $\lambda$-labelings.
Discrete Applied Math. 113.1, 59--72 (2001)

\bibitem{Santoro13}
Flocchini, P., Mans, B., Santoro, N.:
On the exploration of time-varying networks.
Theoretical Computer Science 469, 53--68 (2013)

\bibitem{Garey}
Garey, M., Johnson, D.: 
Computers and Intractability: A guide to the theory of NP-completeness.
W. H. Freeman (1979)

\bibitem{Hopcroft}
Hopcroft, J., Karp R.: An $n^{5/2}$ algorithm for maximum matchings in bipartite graphs.
SIAM Journal on computing 2.4, 225--231 (1973)

\bibitem{Wade11}
Ilcinkas, D., Wade, A.: 
On the power of waiting when exploring public transportation systems.
Prin. of Distr. Sys. Springer, 451--464 (2011)

\bibitem{Wade13}
Ilcinkas D., Wade, A.: 
Exploration of the T-interval-connected dynamic graphs: the case of the ring.
Struct. Info. and Comm. Complexity. Spring, 13--23 (2013)

\bibitem{Karp}
Karp R.: 
Reducibility among combinatorial problems.
In Complexity of Computer Computations. New York: Plenum, 85--103 (1972)

\bibitem{Kuhn10}
Kuhn, F., Lynch, N., Oshman, R.: 
Distributed computation in dynamic networks.
ACM Symp. on Theory of Comp. (2010)

\bibitem{Kuhn11}
Kuhn, F., Oshman, R.: 
Dynamic networks: models and algorithms.
ACM SIGACT News 42.1, 82--96 (2011)

\bibitem{Mans13}
Mans, B., Mathieson, L.: 
On the treewidth of dynamic graphs.
COCOON, 349--360 (2013)

\bibitem{Michail14}
Michail, O., Spirakis, P.:
Traveling salesman problems in temporal graphs.
MFCS (2014), in press.

\bibitem{Wagner99}
Wagner, A., Lindenbaum, M., Bruckstein, A.:
Distributed covering by ant-robots using evaporating traces.
IEEE Trans. on Robo. and Autom. 15.5, 918--933 (1999)

\bibitem{Xuan03}
Xuan, B., Ferreira, A., Jarry, A.: 
Computing shortest, fastest, and foremost journeys in dynamic networks.
IJ Found. Comp. Sci. 14.02, 267--285 (2003)

\end{thebibliography}
\end{document}